\documentclass{cmslatex}
\usepackage[paperwidth=7in, paperheight=10in, margin=.875in]{geometry}
 \usepackage[backref,colorlinks,linkcolor=red,anchorcolor=green,citecolor=blue]{hyperref}
\usepackage{amsfonts,amssymb}
\usepackage{amsmath,bm}
\usepackage{graphicx}
\usepackage{cite}
\usepackage{enumerate}
\renewcommand{\theequation}{\arabic{section}.\arabic{equation}}

\usepackage{mathrsfs}
   \allowdisplaybreaks

\usepackage[utf8x]{inputenc}
\usepackage{esint}
\usepackage[english,french,german,italian,russian]{babel}
\usepackage{appendix}
\def\Fint{\slashint\hspace{-0.5em}\slashint\hspace{-0.5em}\slashint}

\DeclareRobustCommand{\cyrins}[1]{%
  \begingroup\fontfamily{erewhon-TLF}%
  \foreignlanguage{russian}{#1}%
  \endgroup
}

\DeclareMathOperator{\Tr}{Tr}
\DeclareMathOperator{\Cl}{Cl}

\DeclareMathOperator{\Div}{div}

\DeclareMathOperator{\D}{d}
\DeclareMathOperator{\I}{Im}
\DeclareMathOperator{\R}{Re}

\def\eor{\hfill$ \square$}

\usepackage{enumitem}[2011/09/28]
\setenumerate{align=left, leftmargin=0pt,labelsep=.5em, labelindent=0\parindent,listparindent=\parindent,itemindent=*}

\def\Fint{-\!\!\!-\hspace{-1.65em}\iiint}

\date{\today}
\usepackage{colonequals}

\usepackage{tikz}
\tikzset{>=stealth}
\usepackage{pgfplots}
\usepackage{color}
\usepackage[normalem]{ulem}

\begin{document}\selectlanguage{english}

\title{Quantitative spectral analysis of electromagnetic scattering.\ I: $ L^2$ and Hilbert--Schmidt norm bounds}
\author{Yajun Zhou\thanks{Program in Applied and Computational Mathematics, Princeton University, Princeton, NJ 08544, USA; Academy of Advanced Interdisciplinary Studies (AAIS), Peking University, Beijing 100871, P. R. China, (yajunz@math.princeton.edu, yajun.zhou.1982@pku.edu.cn). } }

\maketitle
\begin{abstract}
We perform quantitative spectral analysis  on the Born equation, an integral equation
for electromagnetic scattering that descends from the Maxwell equations. We establish norm bounds  for the Green operator associated with the Born equation, thereby providing numerical tools for error estimates of the Born approximation to light scattering problems.  \end{abstract}


\allowdisplaybreaks
\section{Introduction}There are  wide applications of electromagnetic scattering in astronomy~\cite{LummeBowell,Giese,LeinertGruen}, meteorology~\cite{HulstMeteo,vandeHulst,Aydin}, optical microscopy~\cite{Bruno2003,RICM}, disease diagnosis~\cite{ParkRBC} and thermal therapy~\cite{El-Sayed}. Most of these applications draw on  the quantitative understanding of the scattering patterns arising from the interaction between  (visible or invisible) light   and dielectric media.

In this series of work, we perform quantitative spectral analysis on  electromagnetic scattering and investigate approximate analytic solutions to the light-matter interaction. Extending the qualitative spectral analysis of electromagnetic scattering in our previous work \cite{QualEM} to quantitative studies, we revisit the perturbative solution (Born approximation) to the electromagnetic  scattering problem with an  estimation of  its error bound (in Paper I), and propose a non-perturbative alternative that improves the convergence rate of perturbative methods (in Paper II \cite{QuantEM2Evolv}).
\subsection{Born equation for electromagnetic scattering and its qualitative spectral analysis\label{subsec:qual_summary}}We first  recapitulate some results from \cite{QualEM}, regarding the robust solution to electromagnetic scattering problems for arbitrarily shaped homogeneous and isotropic dielectric media.

The electromagnetic scattering problem is modeled by the \textit{Born equation} (see \cite[Kap.~VII]{BornOptik} or \cite[\S13.6.1]{BornWolf} for the namesake), an integro-differential equation in the form
of\begin{align}(1+\chi)\bm E(\bm r)=\bm E_{\mathrm{inc}} (\bm r)+\chi\nabla\times \nabla\times\iiint_V\frac{\bm E(\bm r') e^{-ik|\bm r-\bm r' |}}{4\pi|\bm r-\bm r' |}\D^3 \bm r',\quad\bm r\in V,\label{eq:Born_eqn}\end{align}where $\bm E(\bm r),\bm r\in V$ is the total electric field established inside the dielectric volume $V$   in response to the incident field $ \bm E_\mathrm{inc}(\bm r)$ [being the spatial part of  a time-harmonic electric field  $\bm E_{\mathrm{inc}} (\bm r,t)=\bm E_{\mathrm{inc}} (\bm r)e^{i\omega t} $]. The dielectric susceptibility is $\chi$, and the wavelength of the incident beam is $2\pi/k=2\pi c/\omega$, with $c$ being the speed of light. The Born equation is an exact consequence from the Maxwell equations.  (See  \cite[\S13.6.1]{BornWolf}, for a formal derivation and  physical interpretations of the Born equation. See \cite[\S2.6]{Potthast2001} or \cite[\S2]{Kirsch2007}, for mathematical proofs that the Born equation descends from the Maxwell equations and scattering boundary conditions.)

As in \cite{QualEM},  we will continue using the terminology ``arbitrarily shaped dielectric'' with three additional geometric qualifications (unless explicitly relaxed otherwise): the dielectric volume $V\subset \mathbb R^3$ is bounded and open, the dielectric boundary $\partial V$ is smooth, and the exterior volume $ \mathbb R^3\smallsetminus(V\cup\partial V)$ is connected.
Following the convention in \cite{QualEM}, we will also use the abbreviation of the Born equation \eqref{eq:Born_eqn} as $ \hat{\mathscr B}\bm E=(\hat I-\chi\hat{\mathscr G})\bm E=\bm E_\mathrm{inc} $. Here, $\hat{\mathscr B}$ is called the \textit{Born operator}, and $\hat{\mathscr G}$ the \textit{Green operator}.

The spectral analysis in \cite[\S3]{QualEM} leads us to the following solvability result when a transverse beam of incident light  shines on the dielectric medium:\begin{quote}\textit{If  the dielectric susceptibility $\chi$ satisfies $\I\chi\leq0 $ and $\chi\neq-2$,\footnote{The critical susceptibility $ \chi=-2$ represents  shape-independent optical resonance where the  energy enhancement ratio $ \iiint_V|\bm E(\bm r)|^2\D^3\bm r/ \iiint_V|\bm E_{\mathrm{inc}}(\bm r)|^2\D^3\bm r$ is unbounded as one exhausts all non-vanishing incident fields $\iiint_V|\bm E_{\mathrm{inc}}(\bm r)|^2\D^3\bm r\neq0 $. This contrasts with other geometry-dependent enhancement effects in the electrostatic approximation \cite{Geom2D,3DconvEPS} to electromagnetic scattering. Using different settings of function spaces [requiring square-integrability of both $ \bm E(\bm r)$ and $\nabla\times \bm E(\bm r) $], Hsiao--Kleinman \cite[Sect.~VI]{HsiaoKleinman}  and Costabel--Darrigrand--Sakly \cite[Corollary 3.3]{CostabelDarrigrandSakly2012}  arrived at similar conclusions about the presence of shape-independent resonance at  $ \chi=-2$ for 3-dimensional electromagnetic scattering.} then the Born equation $ \hat{\mathscr B}\bm E=(\hat I-\chi\hat{\mathscr G})\bm E=\bm E_\mathrm{inc} $ has a unique solution for every incident beam that satisfies both the energy condition  $ \iiint_V|\bm E_\mathrm{inc}(\bm r )|^{2}\D^3\bm r<+\infty$ and the transversality condition $ \nabla\cdot\bm E_\mathrm{inc}(\bm r)=0,\bm r\in V$. Moreover,  except when $ \chi$ belongs to  a countable set of singularities\footnote{Rahola \cite{Rahola} and Budko--Samokhin \cite{Budko} suggested potential existence of a continuum in the spectrum of the volume integral operator for electromagnetic scattering by a dielectric. The problems in  their conjectural arguments were noted by the present author \cite{ZhouThesis}, Costabel--Darrigrand--Sakly \cite{CostabelDarrigrandSakly2012} and Zouros--Budko \cite{ZourosBudko2012}.} (``optical resonance modes'') in the entire complex $\chi$-plane, the solution  $\bm E=(\hat I-\chi\hat{\mathscr G})^{-1}\bm E_\mathrm{inc} $  responds robustly [with $ (\hat I-\chi\hat{\mathscr G})^{-1}$ being a bounded linear operator] to the incident field $ \bm E_{\mathrm{inc}}$.
}\end{quote} In more technical terms, we have demonstrated in \cite[\S\S2--3]{QualEM} that, with the exception of two points $\lambda\in\{0,-1/2\} $, all the singularities of the operator $(\lambda\hat I-\hat{\mathscr G})^{-1} $ are eigenvalues of the Green operator $\hat{\mathscr G}$ lying in the open lower-half complex  $\lambda $-plane ($\I\lambda<0 $).
\subsection{Some tasks in quantitative spectral analysis of electromagnetic scattering}To deduce approximate analytic formulae that describe the scattering patterns in the perturbative regime $ |\chi|\ll1$, the Born approximation is a standard approach. Formally, the Born approximation is based on the operator series of Neumann type $ \hat{\mathscr B}^{-1}=(\hat I-\chi\hat{\mathscr G})^{-1}=\hat I+\chi\hat{\mathscr G}+\chi^2\hat{\mathscr G}^2+\cdots$. The truncation of the Neumann series for the Born equation (sometimes also referred to as the \textit{Born series}) \begin{align}\bm E=(\hat I-\chi\hat{\mathscr G})^{-1}\bm E_\mathrm{inc} =\bm E_\mathrm{inc}+\chi\hat{\mathscr G}\bm E_\mathrm{inc}+\chi^2\hat{\mathscr G}^2\bm E_\mathrm{inc}+\cdots\end{align} up to the term of order $O( \chi^m)$ leads to the $m$-th \textit{Born approximation} to the light scattering problem. Now, we use $ \Vert\cdot\Vert_{L^2(V;\mathbb C^3)}$ to represent the norm of an element in the Hilbert space $ L^2(V;\mathbb C^3)$ as well as the norm of a linear operator acting on the  Hilbert space $ L^2(V;\mathbb C^3) $. According to a standard argument in Banach algebra, the error estimate for the $m$-th {Born approximation} can be put as \begin{align}\left\Vert\bm E-\left[\hat I+\sum_{s=1}^m(\chi\hat{\mathscr G})^s \right]\bm E_\mathrm{inc}  \right\Vert_{L^2(V;\mathbb C^3)}\leq\frac{|\chi|^{m}\Vert\hat{\mathscr G}\Vert_{L^2(V;\mathbb C^3)}^{m}\Vert\bm E_\mathrm{inc}  \Vert_{L^2(V;\mathbb C^3)}}{1-|\chi|\Vert \hat{\mathscr G}\Vert_{L^2(V;\mathbb C^3)}},\label{eq:Born_approx_err_est}\end{align}
so long as  $ |\chi|<1/\Vert \hat{\mathscr G}\Vert_{L^2(V;\mathbb C^3)}$. Therefore, an estimate of the operator norm  $ \Vert \hat{\mathscr G}\Vert_{L^2(V;\mathbb C^3)}$ not only gives a quantitative meaning for  ``the perturbative regime $ |\chi|\ll1$'', but also provides numerical guidance for the error bounds in the Born approximation.

In Paper I, our  task in quantitative spectral analysis will revolve around the estimate of the operator norm   $ \Vert \hat{\mathscr G}\Vert_{L^2(V;\mathbb C^3)}$  based on the geometric shape of the dielectric volume~$V$ and the incident wavelength $ 2\pi/k$. A technical summary of the quantitative  results in the first task will appear in \S\S\ref{subsec:L2norm_bd_summary}--\ref{subsec:HSnorm_bd_summary}.

Practically speaking, Born approximation is usually not carried to the second order or higher, due to analytic difficulties in evaluating the expressions  $ \hat{\mathscr G}^m\bm E_\mathrm{inc}$ for $ m\geq2$. However, the first Born approximation $ \bm E=(\hat I-\chi\hat{\mathscr G})^{-1}\bm E_\mathrm{inc} \approx\bm E_\mathrm{inc}+\chi\hat{\mathscr G}\bm E_\mathrm{inc}$ may not provide an analytic formula with sufficient accuracy for applications to scenarios with moderately large dielectric susceptibilities $\chi$.
In Paper II, we will devise a non-perturbative solution whose computational complexity is comparable to that of the first Born approximation, but retains acceptable numerical  accuracy in a much wider range of susceptibilities $\chi$.   \section{Statement of results}
\subsection{$ L^2$-norm bounds\label{subsec:L2norm_bd_summary}}In \cite[\S2.1]{QualEM}, we showed that the Green operator $ \hat{\mathscr G}\colon L^2(V;\mathbb C^3)\longrightarrow L^2(V;\mathbb C^3)$ given by \begin{align}(\hat{\mathscr G}\bm E)(\bm r)\colonequals \nabla\times \nabla\times\iiint_V\frac{\bm E(\bm r') e^{-ik|\bm r-\bm r' |}}{4\pi|\bm r-\bm r' |}\D^3 \bm r'-\bm E(\bm r),\quad\bm r\in V\end{align}is a well-defined continuous linear transformation for   a bounded and open volume $V\Subset\mathbb R^3 $, provided that  the partial derivatives in the operation $ \nabla\times\nabla\times $ are interpreted in the distributional sense. It has also been shown that the Green operator can be decomposed into the sum of a Hermitian operator $ \hat{\mathscr G }-\hat\gamma\colon L^2(V;\mathbb C^3)\longrightarrow L^2(V;\mathbb C^3)$ and  a Hilbert--Schmidt operator   $ \hat\gamma\colon L^2(V;\mathbb C^3)\longrightarrow L^2(V;\mathbb C^3)$. Here,  the Hermitian operator $ \hat{\mathscr G }-\hat\gamma\colon L^2(V;\mathbb C^3)\longrightarrow L^2(V;\mathbb C^3)$ is defined as \begin{align}((\hat{\mathscr G}-\hat\gamma)\bm E)(\bm r)\colonequals \nabla\left[\nabla\cdot\iiint_V\frac{{\bm E(\bm r')} }{4\pi|\bm r-\bm r' |}\D^3 \bm r'\right],\label{eq:stat_induction}\end{align}and its spectrum is a subset of the closed interval $ [-1,0]$. The Hilbert--Schmidt operator  $ \hat \gamma\colon L^2(V;\mathbb C^3)\longrightarrow L^2(V;\mathbb C^3)$ given  by\begin{align}(\hat\gamma\bm E)(\bm r)\colonequals \nabla\times \nabla\times\iiint_V\frac{{\bm E(\bm r')} (e^{-ik|\bm r-\bm r' |}-1)}{4\pi|\bm r-\bm r' |}\D^3 \bm r'\label{eq:gamma_def}\end{align}exhibits $ O(|\bm r-\bm r'|^{-1})$ asymptotic behavior in the integral kernel. The adjoint $ \hat\gamma^*$ of the non-Hermitian  operator $ \hat\gamma$ is constructed by replacing $ e^{-ik|\bm r-\bm r' |}$ with $ e^{ik|\bm r-\bm r' |}$ in \eqref{eq:gamma_def}.
For convenience, we set $\hat\gamma_C\equiv\R\hat\gamma\colonequals(\hat\gamma+\hat\gamma^*)/2$ and $ \hat\gamma_S\equiv-\I\hat\gamma\colonequals-(\hat\gamma-\hat\gamma^*)/(2i)$, so that we have $\hat \gamma =\hat\gamma_C-i\hat\gamma_S\colon L^2(V;\mathbb C^3)\longrightarrow L^2(V;\mathbb C^3)$ where
\begin{align}(\hat{\gamma}_C\bm E)(\bm r)={}&\nabla\times \nabla\times \iiint_V\frac{\bm E(\bm r' )[\cos(k|\bm r-\bm r'|)-1]}{4\pi|\bm r-\bm r' |}  \D^3 \bm r',\label{eq:gamma_C}\\ (\hat{\gamma}_S\bm E)(\bm r)={}&\nabla\times \nabla\times \iiint_V\frac{\bm E(\bm r' )\sin(k|\bm r-\bm r'|)}{4\pi|\bm r-\bm r' |}  \D^3 \bm r'\label{eq:gamma_S}\end{align}define two Hermitian   compact operators.

The norm bounds for the Green operator $ \hat{\mathscr G}$ is stated in the theorem below. \begin{theorem}[{$ L^2$-norm bounds}]
\label{thm:L2_op_norm}For a bounded open dielectric volume $V$ with smooth boundary $ \partial V$, we have $ \Vert\hat{\mathscr G}\Vert\leq\Vert\R\hat{\mathscr G}\Vert+\Vert\I\hat{\mathscr G}\Vert$ for
\begin{small}\begin{align}
&\max_{\lambda\in\sigma(\hat{\mathscr G})}|\R\lambda|\leq\Vert\R\hat {\mathscr G}\Vert\notag\\\leq{}&1+\min\left\{ \frac{3 }{5\pi}\left(k^3\iiint_V\D^3\bm r\right)^{2/3},1+\frac{4}{5} k R_{V}+\frac{1}{2\pi} \left(  \sqrt[3]{k R_{V}+\frac{1}{2}}+\frac{4}{9 \sqrt[3]{k R_{V}+\frac{1}{2}}} \right)\right\},\label{eq:Re_gamma_norm_bd}
\\&-\min_{\lambda\in\sigma(\hat{\mathscr G})}\I\lambda=\max_{\lambda\in\sigma(\hat{\mathscr G})}|\I\lambda|\leq\Vert\I\hat {\mathscr G}\Vert=\Vert\I\hat \gamma\Vert=\Vert\hat\gamma_S\Vert\notag\\\leq{}&\min\left\{\frac{k^3}{4\pi}\iiint_V\D^3\bm r,\frac{11}{45}( kR_V)^3, \frac{3 }{5\pi}\left(k^3\iiint_V\D^3\bm r\right)^{2/3}, \frac{9}{16}kR_V \right\},\label{eq:Im_gamma_norm_bd}
\end{align}\end{small}
and
\begin{small}\begin{align} \Vert \hat{ \gamma}\Vert_{L^2(V;\mathbb C^3)}\leq{}&\min\left\{ \frac{3 }{5\pi}\left(k^3\iiint_V\D^3\bm r\right)^{2/3},2+\frac{4}{5} k R_{V}+\frac{1}{2\pi} \left(  \sqrt[3]{k R_{V}+\frac{1}{2}}+\frac{4}{9 \sqrt[3]{k R_{V}+\frac{1}{2}}} \right)\right\}\notag\\&+\min\left\{\frac{k^3}{4\pi}\iiint_V\D^3\bm r,\frac{11}{45}( kR_V)^3, \frac{3 }{5\pi}\left(k^3\iiint_V\D^3\bm r\right)^{2/3}, \frac{9}{16}kR_V \right\},\label{eq:gamma_norm_bd_statement}
\end{align}\end{small}where  $ \sigma(\hat{\mathscr G})$ stands for the spectrum of the bounded linear operator $\hat{\mathscr G}\colon L^2(V;\mathbb C^3)\longrightarrow L^2(V;\mathbb C^3) $ and  $R_V\colonequals \min_{\bm r\in\mathbb R^3}\max_{\bm r'\in V\cup\partial V}|\bm r'-\bm r|$ is the minimal radius of all the circumscribed spheres. \end{theorem}

\subsection{Hilbert--Schmidt norm bounds\label{subsec:HSnorm_bd_summary}}
In \cite[\S\S2--3]{QualEM}, we performed spectral analysis for the Green operator hitting on a  Hilbert subspace $ \hat{\mathscr G}\colon \Phi(V;\mathbb C^3)\longrightarrow \Phi(V;\mathbb C^3) $. Here, the Hilbert subspace $\Phi(V;\mathbb C^3)=\Cl(C^\infty(V;\mathbb C^3)\cap\ker(\nabla\cdot)\cap L^2(V;\mathbb C^3))$  is the $ L^2$-closure for the totality of smooth, divergence-free and square-integrable $ \mathbb C^3$-valued vector fields.
Conventionally, the Hilbert space $\Phi(V;\mathbb C^3)$ is denoted by  $ H(\Div0 ,V)$ \cite[p.~215]{DautrayLionsVol3}.

The transversality condition $ \nabla\cdot\bm E=0$ is a physical requirement from one of the Maxwell equations. Thus, points in the spectrum $ \sigma^{\Phi}(\hat{\mathscr G})$ of the bounded linear operator  $ \hat{\mathscr G}\colon \Phi(V;\mathbb C^3)\longrightarrow \Phi(V;\mathbb C^3) $ correspond to physically meaningful \textit{optical resonance modes} in electromagnetic scattering.

In \cite[\S3.3]{QualEM}, we proved the convergence of a  spectral series:\footnote{The sum over $ \lambda\in\sigma^\Phi(\hat{\mathscr G})$ respects multiplicities of eigenvalues. See \eqref{eq:muHS} and Corollary \ref{cor:HS_Mie} below.}\begin{align}\label{eq:conv_spectral_series}\sqrt{\vphantom{\sum_{a}^{_a}}\smash{\sum_{\lambda\in\sigma^\Phi(\hat{\mathscr G})}^{\phantom{a}}|\lambda|^2|1+2\lambda|^4}}\leq\Vert\hat{\mathscr G }(\hat I+2\hat{\mathscr G})^2\Vert_2\colonequals \sqrt{\sum_{s=1}^\infty\Vert\hat{\mathscr G }(\hat I+2\hat{\mathscr G})^2\bm e_s\Vert_{L^2(V;\mathbb C^3)}^2}<+\infty,\end{align}where $ \{\bm e_s|s=1,2,\dots\}$ is any complete set of orthonormal basis for the Hilbert space $ \Phi(V;\mathbb C^3)$. Here in \eqref{eq:conv_spectral_series},  the lower bound for the Hilbert--Schmidt norm  $ \Vert\hat{\mathscr G }(\hat I+2\hat{\mathscr G})^2\Vert_2$ comes from the Schur--Weyl inequality \cite[p.~8]{Simon2005}, while the finite upper bound  builds on three facts:\begin{enumerate}[label={(\arabic*)}, ref=(\alph*), widest=a]\item We have   $\Vert\hat{\mathscr G}-\hat\gamma\Vert\leq1 $, $\Vert\hat I+\hat{\mathscr G}-\hat\gamma \Vert\leq1 $ and $ \Vert\hat\gamma\Vert\leq\Vert\hat\gamma\Vert_2<+\infty$;
\item
There is an inequality $ \Vert(\hat I+2\hat{\mathscr G})^2\hat{\mathscr G }\Vert_2\leq\Vert(\hat I+2\hat{\mathscr G}-2\hat\gamma)^2(\hat{\mathscr G}-\hat\gamma)\Vert_2+\Vert\hat\gamma\Vert_2(13+12\Vert\hat\gamma\Vert+4\Vert\hat\gamma\Vert^2)$, due to $ \Vert\hat A\hat B\Vert_2\leq\Vert\hat A\Vert_2\Vert\hat B\Vert$ and $\Vert\hat B\hat A\Vert_2\leq\Vert\hat A\Vert_2\Vert\hat B\Vert $ \cite[p.~218]{ReedSimonVolI1980};
\item The Hilbert--Schmidt norm $\Vert(\hat I+2\hat{\mathscr G}-2\hat\gamma)^2(\hat{\mathscr G}-\hat\gamma)\Vert_2$ is bounded by  a convergent multiple surface integral  when $ \partial V$ is smooth:\begin{align}&
\Vert(\hat I+2\hat{\mathscr G}-2\hat\gamma)^2(\hat{\mathscr G}-\hat\gamma)\Vert_2\leq\sqrt{\sum_{s=1}^\infty\Vert(\hat I+2\hat{\mathscr G}-2\hat\gamma)^2\bm f_s\Vert_{L^2(V;\mathbb C^3)}^2}\notag\\={}&\sqrt{\vphantom{\sum_{a}^{_a}}\smash[b]{\varoiint_{\partial V}\D S_1\varoiint_{\partial V}\D S_2\varoiint_{\partial V}\D S_3\varoiint_{\partial V}\D S_4\prod_{j=1}^4\left[ \frac{\bm n_{j }\cdot(\bm r_{j}-\bm r_{j+1\bmod 4}) }{2\pi|\bm r_{j}-\bm r_{j+1\bmod 4}|^{3}} -\frac{1}{\varoiint_{\partial V}\D S_{j}}\right]}},\label{eq:quad_surf_int}
\end{align}where    $ \{\bm f_s|s=1,2,\dots\}$ exhausts all the eigenvectors corresponding  to the non-zero eigenvalues of the polynomially compact Hermitian operator $\hat{\mathscr G}-\hat \gamma\colon \Phi(V;\mathbb C^3)\longrightarrow\Phi(V;\mathbb C^3) $.\end{enumerate}

 In \cite{QualEM}, we did not elaborate on the bound estimates for $ \Vert\hat\gamma\Vert_2$ in terms of the shape of $V$ and the size of the wavelength $ 2\pi/k$. This will be given in the next theorem.\begin{theorem}[Hilbert--Schmidt bounds]\label{thm:HS_norm}For arbitrarily shaped bounded open dielectric volume $V$, we have
\begin{align} \Vert\hat\gamma\Vert_2\leq\min\left\{\frac{1}{3}\left(k^3\iiint_V\D^3\bm r\right)^{2/3},\frac34(kR_V)^2\right\},\label{eq:gamma_HS_bd_statement}
\end{align}where    $R_V$ is the minimal radius of all the circumscribed spheres.\end{theorem}

\section{Analytic preparations}In this section, we shall   briefly overview the $ L^2$-formulation for light scattering problems in \S\ref{sec:L2_form}, and discuss the exact solutions to spherical scatterers in \S\ref{subsec:Mie_eigen}, to pave the way for the proofs of  Theorems~\ref{thm:L2_op_norm} and \ref{thm:HS_norm} in \S\ref{sec:norm_bds}.
\subsection{$L^2$-formulation of electromagnetic scattering\label{sec:L2_form}}In \cite[\S2]{QualEM}, we defined the Born operator   $\hat{\mathscr B}\colon L^2(V;\mathbb C^3)\longrightarrow L^2(V;\mathbb C^3)$ by interpreting the partial derivatives on the right-hand side of the following equation
\begin{align}(\hat{\mathscr B}\bm E)(\bm r)\colonequals (1+\chi)\bm E(\bm r)-\chi\nabla\times \nabla\times\iiint_V\frac{\bm E(\bm r') e^{-ik|\bm r-\bm r' |}}{4\pi|\bm r-\bm r' |}\D^3 \bm r',\quad\bm r\in V\end{align} in the distributional sense. This was achieved by first defining    $\hat{\mathscr B}\colon C^\infty_0(V;\mathbb C^3)\longrightarrow  C^\infty(V;\mathbb C^3)\cap L^2(V;\mathbb C^3)$ with classical derivatives [where $ C^\infty_0(V;\mathbb C^3)$ is the totality of smooth and compactly supported $ \mathbb C^3$-valued vector fields defined in $V$] and then proving that there is a continuous extension to
$\hat{\mathscr B}\colon L^2(V;\mathbb C^3)\longrightarrow L^2(V;\mathbb C^3)$ when the distributional derivatives are taken.

In the literature concerning numerical solutions of the light scattering problem,  the Born equation is usually formulated via the digitized Green function approach (also known as the discrete dipole approximation) as follows \cite{Holt,Rahola,Shifrin,Bladel,Mishchenko,Goedecke,Yurkina,Budko}:\begin{align}(\hat{\mathscr B'}\bm E)(\bm r)\colonequals {}&\bm E(\bm r)-\chi\left[ -\frac{\bm E(\bm r )}{3}+k^2 \iiint_V\frac{\bm E(\bm r' )e^{-ik|\bm r-\bm r' |}}{4\pi|\bm r-\bm r' |}  \D^3 \bm r'+\right.\notag\\{}&\left.+\lim_{\varepsilon\to0^+}\iiint_{V\smallsetminus O(\bm r,\varepsilon)} \nabla\nabla\frac{e^{-ik|\bm r-\bm r' |}}{4\pi|\bm r-\bm r' |} \cdot\bm E(\bm r' ) \D^3\bm r'\right]\label{eq:GreenBorn}.\end{align} where the Hessian matrix reads\begin{align}\label{eq:Hess_prop}&\nabla\nabla\frac{e^{-ik|\bm r-\bm r' |}}{4\pi|\bm r-\bm r' |}\notag\\=&\frac{k^2 e^{-ik|\bm r-\bm r' |}}{4\pi|\bm r-\bm r' |}\left[-\left( \frac{i}{k|\bm r-\bm r' |}+\frac{1}{k^2|\bm r-\bm r' |^2} \right)\mathbf1+\left(-1+ \frac{3i}{k|\bm r-\bm r' |}+\frac{3}{k^2|\bm r-\bm r' |^2} \right)\hat {\bm r}\hat{\bm r}^\mathrm{T}\right],\end{align}with $\hat {\bm r}\hat{\bm r}^\mathrm{T}=(\bm r-\bm r')(\bm r-\bm r')^{\mathrm T}/|\bm r-\bm r'|^2 $.

In the following, we shall verify  that the  two foregoing definitions of the Born operator agree with each other when $\bm E\in L^{2}(V;\mathbb C^3) $, that is,   $\hat{\mathscr B}=\hat{\mathscr B}'\colon  L^{2}(V;\mathbb C^3)\longrightarrow   L^2(V;\mathbb C^3)$.
 It is straightforward  to check that    $\hat{\mathscr B}=\hat{\mathscr B}'\colon C^\infty_0(V;\mathbb C^3)\longrightarrow  C^\infty(V;\mathbb C^3)\cap L^2(V;\mathbb C^3)$. To show that the two definitions (distributional derivative in $\hat{\mathscr B}$ and digitized Green function for $\hat{\mathscr B}'$) of the Born operator agree for all square-integrable inputs $\bm E\in L^2(V;\mathbb C^3)  $, it would suffice to give a direct proof  that $\hat{\mathscr B}'\colon L^2(V;\mathbb C^3)\longrightarrow L^2(V;\mathbb C^3) $ defines a continuous linear mapping, as demonstrated below.
\begin{proposition}
The linear mapping $\hat{\mathscr B}'=\hat I-\chi(\hat{\mathscr G}_0+\hat{\mathscr G}_1+\hat{\mathscr G}_2+\hat{\mathscr G}_3)\colon L^2(V;\mathbb C^3)\longrightarrow L^2(V;\mathbb C^3)$  is continuous, where \begin{align}(\hat{\mathscr G}_0\bm E)(\bm r)\colonequals &-\frac{\bm E(\bm r )}{3},\notag\\(\hat{\mathscr G}_1\bm E)(\bm r)\colonequals &\iiint_V\frac{k^2 e^{-ik|\bm r-\bm r' |}}{4\pi|\bm r-\bm r' |}\left(\mathbf1-\hat {\bm r}\hat{\bm r} ^\mathrm{T}\right)  \bm E(\bm r' )\D^3 \bm r'=\iiint_V\hat \Gamma_1(\bm r,\bm r') \bm E(\bm r' )\D^3 \bm r',\notag\\(\hat{\mathscr G}_2\bm E)(\bm r)\colonequals &-i\iiint_V\frac{k e^{-ik|\bm r-\bm r' |}}{4\pi|\bm r-\bm r' |^2}\left(\mathbf1-3\hat {\bm r}\hat{\bm r} ^\mathrm{T}\right)  \bm E(\bm r' )\D^3 \bm r'=\iiint_V\hat \Gamma_2(\bm r,\bm r') \bm E(\bm r' )\D^3 \bm r',\notag\\(\hat{\mathscr G}_3\bm E)(\bm r)\colonequals &-\lim_{\varepsilon\to0^+}\iiint_{V\smallsetminus O(\bm r,\varepsilon)}\frac{ e^{-ik|\bm r-\bm r' |}}{4\pi|\bm r-\bm r' |^3}\left(\mathbf1-3\hat {\bm r}\hat{\bm r} ^\mathrm{T}\right)  \bm E(\bm r' )\D^3 \bm r'\notag\\={}&\lim_{\varepsilon\to0^+}\iiint_{V\smallsetminus O(\bm r,\varepsilon)}\hat \Gamma_3(\bm r,\bm r') \bm E(\bm r' )\D^3 \bm r'.\end{align}
\end{proposition}\begin{proof}We claim that there exist four finite numbers $G_m<+\infty,m=0,1,2,3 $ such that $\Vert\hat{\mathscr G}_m\bm E\Vert_{L^2(V;\mathbb C^3)} \leq G_m\Vert\bm E\Vert_{L^2(V;\mathbb C^3)},m=0,1,2,3$ for all $\bm E\in L^2(V;\mathbb C^3) $. Once this can be shown, we can further arrive at $\Vert\hat{\mathscr B'}\bm E
\Vert_{L^2(V;\mathbb C^3)}\leq [1+|\chi|(G_0+G_1+G_2+G_3)]\Vert\bm E\Vert_{L^2(V;\mathbb C^3)},\forall \bm E\in {L^2(V;\mathbb C^3)}  $.

 Here, it is obvious that one can choose $G_0=1/3<+\infty $.

 To show that the constant $G_1 +G_2<+\infty $, we need the Young inequality \cite[p.~271]{SteinDiff} for convolutions $h(\bm r)=\iiint_{\mathbb R^3}f(\bm r' )g(\bm r-\bm r')  \D^3\bm r'$:\begin{align}\iiint_{\mathbb R^3}|h(\bm r ) |^2  \D^3\bm r \leq\left(\iiint_{\mathbb R^3}|g(\bm r' )|\D^3 \bm r'\right)^2\iiint_{\mathbb R^3}|f(\bm r ) |^2  \D^3\bm r.\end{align}
In the current context, we may pick any $\bm u\in \{\bm e_x,\bm e_y,\bm e_z \}$, and apply the Young inequality to \begin{align}f(\bm r)=\begin{cases}\bm u\cdot\bm E(\bm r), & \bm r\in V \\
0, & \bm r\notin V \\
\end{cases},\quad g(\bm r)=\begin{cases}e^{-ik|\bm r |}|\bm r|^{-m}, & |\bm r|<2\max_{\bm r',\bm r''\in V\cup\partial V}|\bm r'-\bm r''| \\
0, & |\bm r|\geq2\max_{\bm r',\bm r''\in V\cup\partial V}|\bm r'-\bm r''| \\
\end{cases} , \end{align}so as to confirm that all the matrix elements of $ \hat\Gamma_m(\bm r,\bm r'),m=1,2$ contribute to a continuous linear mapping.

To show that $G_3<+\infty $, we need the Calder\'on--Zygmund theory in harmonic analysis that deals with singular integral operators \cite{CZ1952}. Specifically, the convolution kernel $K(\bm r,\bm r') $ in question (\textit{i.e.}~any matrix element of $\hat \Gamma_3(\bm r,\bm r') $) qualifies as a singular integral operator of Calder\'on--Zygmund
type (see \cite[p.~293]{Stein}, also \cite[p. 39]{SteinDiff}) ---  for some finite positive constant $A<+\infty$, the convolution kernel $K(\bm r,\bm r' ),\bm r\neq \bm r'$ satisfies the following inequalities:\begin{align}\begin{split}|K(\bm r,\bm r' )|\leq &\frac{A}{|\bm r-\bm r' |^d},  \\
|K(\bm r_1,\bm r' )-K(\bm r_2,\bm r' ) |\leq& \frac{A|\bm r_1-\bm r_2 |}{|\bm r_1-\bm r' |^{d+1}}  \quad \text{ if } |\bm r_1-\bm r_2 |\leq \frac{|\bm r_1-\bm r' |}2,\\|K(\bm r,\bm r'_1 )-K(\bm r,\bm r' _2) |\leq& \frac{A|\bm r_1'-\bm r_2 '|}{|\bm r_1'-\bm r |^{d+1}}  \quad \text{ if } |\bm r_1'-\bm r_2 '|\leq \frac{|\bm r_1'-\bm r |}2,\end{split}\label{eq:CZ}\end{align}where $d=3 $ is the dimensionality of the space under consideration. Furthermore, the Calder\'on--Zygmund kernels $K(\bm r,\bm r' )$ that arise from any one of the matrix elements of  $\hat \Gamma_3(\bm r,\bm r') $ also satisfy the ``cancellation condition'' (see \cite{Stein}, p.~291 and pp.~305-306, also \cite{SteinDiff}, p. 39)\begin{align}\iiint_{\varepsilon<|\bm r-\bm r' |<N}K(\bm r,\bm r' )   \D^3\bm r'=0,\quad\text{ whenever }0<\varepsilon<N<+\infty.\end{align}
Here, for off-diagonal elements of   $\hat \Gamma_3(\bm r,\bm r') $, the cancellation results from reflection symmetry, which leads to
\begin{align}\label{eq:CancRefSym}\iiint_{\varepsilon^2<X^2+Y^2+Z^2<N^2}\frac{g(|\bm R| )}{|\bm R|^2}  XY \D X  \D Y  \D Z=0,\,\text{etc.}\end{align}
 where $g(|\bm R|), |\bm R|=|\bm r-\bm r'|=\sqrt{X^2+Y^2+Z^2}$ is a generic spherically symmetric function; for the diagonal elements of    $\hat \Gamma_3(\bm r,\bm r') $, the cancellation results from rotation symmetry, \textit{i.e.}~we have\begin{align}
\iiint_{\varepsilon^2<X^2+Y^2+Z^2<N^2}\frac{g(|\bm R| )}{|\bm R|^2} X^2\D X  \D Y  \D Z=\frac13\iiint_{\varepsilon^2<X^2+Y^2+Z^2<N^2}g(|\bm R| ) \D X  \D Y  \D Z,\,\text{etc.}\label{eq:cancRotSym}\end{align}It is a standard result in harmonic analysis that a Calder\'on--Zygmund kernel satisfying the cancellation condition induces a bounded linear operator from $ L^2(\mathbb R^3;\mathbb C^3)$ to $L^2(\mathbb R^3;\mathbb C^3)$, which  necessarily specializes to a continuous mapping  from $ L^2(V;\mathbb C^3)$ to $L^2(V;\mathbb C^3)$. This shows that there exists a finite constant $ G_3<+\infty$, such that $\Vert\hat{\mathscr G_3}\bm E\Vert_{L^2(V;\mathbb C^3)} \leq G_3\Vert\bm E\Vert_{L^2(V;\mathbb C^3)}$ for all $\bm E\in L^2(V;\mathbb C^3)$. \end{proof}

In the light of this, we can say that the Born operator    $\hat{\mathscr B}\colon L^2(V;\mathbb C^3)\longrightarrow L^2(V;\mathbb C^3)$ defined by distributional derivatives in   \cite[\S2]{QualEM}  is mathematically equivalent to the formulation of discrete dipole approximation (\textit{i.e.}~digitized Green function approach)  $\hat{\mathscr B}'\colon L^2(V;\mathbb C^3)\longrightarrow L^2(V;\mathbb C^3)$ in the physical  literature on light scattering. Thus, the efforts in the proposition above not only  demonstrate the relevance of our previous studies \cite{QualEM} to the digitized Green function approach (discrete dipole approximation), but also  provide a logical foundation for some norm estimates in \S\ref{sec:norm_bds} (in particular, the bounds for the $ L^2$-norm $ \Vert\hat\gamma\Vert$ in \S\ref{subsec:gamma_L2_bd} and the Hilbert--Schmidt norm  $ \Vert\hat\gamma\Vert_{2}$ in \S\ref{subsec:HS_bd}) that are carried out on the Hessian matrix, instead of distributional derivatives.

\subsection{Mie scattering and its spectral analysis\label{subsec:Mie_eigen}}For light scattering by a dielectric sphere $V=O(\mathbf0,R)\colonequals\{\bm r\in\mathbb R^3|0\leq|\bm r|<R\}$ with susceptibility $\chi$, the exact solution is presented by the Mie series (see  \cite{Mie} or \cite[p.~122]{vandeHulst}) for $|\bm r|<R $:
\begin{align}&\bm E(\bm r)=((\hat I-\chi\hat{\mathscr G})^{-1}\bm E_\mathrm{inc})(\bm r)\notag\\=&\sum_{\ell=1}^{\infty}\sum_{m=-\ell}^\ell\left\{\frac{\alpha_{\ell m}\nabla\times[\bm r j_\ell(\sqrt{1+\chi}k|\bm r|)Y_{\ell m}(\theta,\phi)]}{kR\left[j_\ell(\sqrt{1+\chi}kR)\left.\frac{\D[x h_\ell^{(2)}(x)]}{\D x}\right|_{x=kR}- h_\ell^{(2)}(kR)\left.\frac{\D[y j_\ell^{\phantom{(}}(y)]}{\D y}\right|_{y=\sqrt{1+\chi}kR}\right]}+\right.\notag\\&\left.+\frac{i\beta_{\ell m}\nabla\times\nabla\times[\bm r j_\ell(\sqrt{1+\chi}k|\bm r|)Y_{\ell m}(\theta,\phi)]}{kR\left[(1+\chi)j_\ell(\sqrt{1+\chi}kR)\left.\frac{\D[x h_\ell^{(2)}(x)]}{\D x}\right\vert_{x=kR}- h_\ell^{(2)}(kR)\left.\frac{\D[y j_\ell^{\phantom{(}}(y)]}{\D y}\right|_{y=\sqrt{1+\chi}kR}\right]}\right\},\label{eq:MieSeries}\end{align} where the incident beam \begin{align}&\bm E_\mathrm{inc}(\bm r)\notag\\={}&\sum_{\ell=1}^{\infty}\sum_{m=-\ell}^\ell\left\{i\alpha_{\ell m}\nabla\times[\bm r j_\ell(k|\bm r|)Y_{\ell m}(\theta,\phi)]-\beta_{\ell m}\nabla\times\nabla\times[\bm r j_\ell(k|\bm r|)Y_{\ell m}(\theta,\phi)]\right\}\label{eq:Einc_jY}\end{align}is expressed in terms of the  spherical Bessel functions\footnote{We note that $h^{(2)}_\ell(x)\neq0 $ for all real-valued $x$, according to a standard Wro\'nskian argument. }\begin{align}j_\ell(z)={}&(-z)^\ell\left( \frac{\D}{z\D z} \right)^\ell\left( \frac{\sin z}{z} \right)=z^{\ell}\sum_{m=0}^\infty\frac{(-z^{2})^m}{2^mm!(2\ell+2m+1)!!}
,\label{eq:j_ell}\\h^{(2)}_\ell(z)={}&(-z)^\ell\left( \frac{\D}{z\D z} \right)^\ell\left( \frac{ie^{-i z}}{z} \right)
,\label{eq:h2_ell}\end{align}and the standard spherical harmonics $ Y_{\ell m}(\theta,\phi)$ compatible with the spherical coordinate system $ \bm r=|\bm r|(\sin\theta\cos\phi,\sin\theta\sin\phi,\cos\theta)$. In what follows, for $ \ell\in\mathbb Z_{>0},m\in\mathbb Z\cap[-\ell,\ell],\chi\in\mathbb C$, we define\begin{align}\bm f^{\mathrm{TE}}_{\ell m,\chi}(\bm r)\colonequals{}&
\nabla\times[\bm r j_\ell(\sqrt{1+\chi}k|\bm r|)Y_{\ell m}(\theta,\phi)],\\\bm f^{\mathrm{TM}}_{\ell m,\chi}(\bm r)\colonequals{}&\nabla\times\nabla\times[\bm r j_\ell(\sqrt{1+\chi}k|\bm r|)Y_{\ell m}(\theta,\phi)]
\end{align}as the TE and TM wave modes with $ Y_{\ell m}$ symmetry, where TE (resp.\ TM) wave stands for \textit{t}ransverse \textit{e}lectric (resp.\ \textit{m}agnetic) waves satisfying $\bm r\cdot\bm E=0 $ [resp.\ $\bm r\cdot(\nabla\times\bm
E)=0 $]. Accordingly, we normalize the denominators in the Mie series as\footnote{Here, in  \eqref{eq:TE_res}--\eqref{eq:TM_res} below, the factor $ (1+\chi)^{\ell/2}$ is attributed to the fact that $  j_\ell(\sqrt{1+\chi}k|\bm r|)/(1+\chi)^{\ell/2}$  is holomorphic in $ \chi\in\mathbb C$. Consequently, the expressions $ \bm E(\bm r)=((\hat I-\chi\hat{\mathscr G})^{-1}\bm E_\mathrm{inc})(\bm r)$ in   \eqref{eq:MieSeries}, $ \Delta_\ell^{\mathrm {TE}}(\chi)$  in  \eqref{eq:TE_res} and $ \Delta_\ell^{\mathrm {TM}}(\chi)$  in  \eqref{eq:TM_res} are all meromorphic in $\chi$, without branch cut discontinuities.  As $ \chi\to-1$, we have non-vanishing limits $ \Delta_\ell^{\mathrm {TE}}(\chi)\to-\frac{ (kR)^{\ell+1}h_{\ell+1}^{(2)}(kR)}{(2\ell+1)!!}\neq0$ and $ \Delta_\ell^{\mathrm {TE}}(\chi)\to-\frac{(\ell+1) (kR)^{\ell}h_\ell^{(2)}(kR)}{(2\ell+1)!!}\neq0$, so $\chi=-1$ is not a singularity in the Mie scattering (cf.~\cite[Proposition 3.4]{QualEM}). }
\begin{align}\Delta_{\ell,kR}^{\mathrm {TE}}(\chi)\colonequals{}&\frac{j_\ell\left(\sqrt{1+\chi}kR\right)\left.\frac{\D[x h_\ell^{(2)}(x)]}{\D x}\right|_{x=kR}- h_\ell^{(2)}(kR)\left.\frac{\D[y j_\ell^{\phantom{(}}(y)]}{\D y}\right|_{y=\sqrt{1+\chi}kR}}{(1+\chi)^{\ell/2}},\label{eq:TE_res}\\
\Delta_{\ell,kR}^{\mathrm {TM}}(\chi)\colonequals{}&\frac{(1+\chi)j_\ell\left(\sqrt{1+\chi}kR\right)\left.\frac{\D[x h_\ell^{(2)}(x)]}{\D x}\right|_{x=kR}- h_\ell^{(2)}(kR)\left.\frac{\D[y j_\ell^{\phantom{(}}(y)]}{\D y}\right|_{y=\sqrt{1+\chi}kR}}{(1+\chi)^{\ell/2}}.\label{eq:TM_res}\end{align}

\begin{proposition}[Mie resonances]\label{prop:Mie_res}Fix the radius of the spherical scatterer as $R$, and wavelength as $2\pi/k$. If $ \Delta_{\ell,kR}^{\mathrm {TE}}(\chi^{\mathrm{TE}}_{\ell,kR})=0$ [resp.\ $ \Delta_{\ell,kR}^{\mathrm {TM}}(\chi^{\mathrm{TM}}_{\ell,kR})=0$],\footnote{For given $ \chi$ and $kR$, we have $\Delta_{\ell,kR}^{\mathrm {TM}}(\chi)=-\frac{i [\ell (\chi +2)+1]}{(2 \ell+1)kR}+\frac{i \chi\{\ell [-\chi +2\ell  (\chi +2)+4]+3\} kR}{2 (2\ell-1) (2\ell+1) (2\ell+3)}\left[1+O\left(\frac{k^{2}R^2}{\ell}\right)\right]$ as $\ell\to\infty$. Asymptotically, this contributes resonance modes $ \chi_{\ell,kR}^{\mathrm{TM}}\sim-2 - \frac{1}{\ell} \left(1+ \frac{k^2R^2}{\ell}\right)$ that aggregate around a special point $ \chi=-2$ (cf.~the results from \cite[\S3]{QualEM} recapitulated in \S\ref{subsec:qual_summary}).} then we have an eigenequation $\bm f^{\mathrm{TE}}_{\ell m,\chi^{\mathrm{TE}}_{\ell,kR}}=\chi^{\mathrm{TE}}_{\ell,kR}\hat{\mathscr G} \bm f^{\mathrm{TE}}_{\ell m,\chi^{\mathrm{TE}}_{\ell,kR}}$ (resp.\ $\bm f^{\mathrm{TM}}_{\ell m,\chi^{\mathrm{TM}}_{\ell,kR}}=\chi^{\mathrm{TM}}_{\ell,kR}\hat{\mathscr G} \bm f^{\mathrm{TM}}_{\ell m,\chi^{\mathrm{TM}}_{\ell,kR}}$) for every $ m\in\mathbb Z\cap[-\ell,\ell]$. These special values of $ \chi^{\mathrm{TE}}_{\ell,kR}$ and $ \chi^{\mathrm{TM}}_{\ell,kR}$ are referred to as Mie resonances.\end{proposition}\begin{proof}By inspecting \eqref{eq:MieSeries} and \eqref{eq:Einc_jY}, we find\begin{align}\begin{split}&
(\hat I-\chi\hat{\mathscr G})^{-1}\bm f^{\mathrm{TE}}_{\ell m,0}\\={}&\frac{\bm f^{\mathrm{TE}}_{\ell m,\chi}}{ikR\left[j_\ell(\sqrt{1+\chi}kR)\left.\frac{\D[x h_\ell^{(2)}(x)]}{\D x}\right|_{x=kR}- h_\ell^{(2)}(kR)\left.\frac{\D[y j_\ell^{\phantom{(}}(y)]}{\D y}\right|_{y=\sqrt{1+\chi}kR}\right]},\end{split}\label{eq:Mie_TElm}
\end{align}when the denominator on the right-hand side does not vanish. This can be recast into $(\hat I-\chi\hat{\mathscr G})\bm f^{\mathrm{TE}}_{\ell m,\chi}=ikR(1+\chi)^{\ell/2} \Delta_{\ell,kR}^{\mathrm {TE}}(\chi)\bm f^{\mathrm{TE}}_{\ell m,0}$, whose validity extends to all $\chi\in\mathbb C$, by continuity. When $ \chi=\chi^{\mathrm{TE}}_{\ell,kR}$, we have an eigenequation  $ \bm f^{\mathrm{TE}}_{\ell m,\chi}=\chi\hat{\mathscr G}\bm f^{\mathrm{TE}}_{\ell m,\chi}$ as claimed. One can similarly argue for the TM wave modes.  \end{proof}

 Mie's exact solutions offer more than a family of specific examples with explicitly known  spectral structures in the Green operator $ \hat{\mathscr G}$. They will also improve our norm bounds of $ \I\hat{\mathscr G}=-\hat \gamma_S$ (for generically shaped scatterers) in \S\ref{subsec:totalsc}, thanks to the next proposition.
   \begin{proposition}[Eigenmodes for $ \hat\gamma_S$]\label{prop:Mie_gamma_S} With the abbreviations $\bm f^{\mathrm{TE}}_{\ell m}\colonequals \bm f^{\mathrm{TE}}_{\ell m,0}$ and  $\bm f^{\mathrm{TM}}_{\ell m}\colonequals \bm f^{\mathrm{TM}}_{\ell m,0}$,  along with the notations  $ \lambda_{n}^{\mathrm{TE}}(x)\colonequals\int_0^{x}j_n^2(\xi)\xi^2\D \xi,n\in\mathbb Z_{\geq0}$ and $\lambda_{n}^{\mathrm{TM}}(x)=\frac{n+1}{2n+1}\lambda_{n-1}^{\mathrm{TE}}(x)+\frac{n}{2n+1}\lambda_{n+1}^{\mathrm{TE}}(x),n\in\mathbb Z_{>0} $,  we have the following eigenequations:\begin{align}
\hat{\gamma}_S\bm f^{\mathrm{TE}}_{\ell m}={}&\lambda_{\ell}^{\mathrm{TE}}(kR)\bm f^{\mathrm{TE}}_{\ell m},\label{eq:gammaS_TE}\\\hat{\gamma}_S\bm f^{\mathrm{TM}}_{\ell m}={}&\lambda_{\ell}^{\mathrm{TM}}(kR){\bm f^{\mathrm{TM}}_{\ell m}},\label{eq:gammaS_TM}
\end{align}where $ \ell\in\mathbb Z_{>0}$ and $m\in\mathbb Z\cap[-\ell,\ell] $. Moreover, these relations  enumerate all the eigenvectors associated with non-vanishing eigenvalues of the compact Hermitian operator
$\hat\gamma_S\colon L^2( O(\mathbf0,R);\mathbb C^3)\longrightarrow L^2( O(\mathbf0,R);\mathbb C^3)$. \end{proposition}\begin{proof}As we compute \begin{align}
\hat{\mathscr G}\bm f^{\mathrm{TE}}_{\ell m}=\lim_{\chi\to0}\frac{(\hat I-\chi\hat{\mathscr G})^{-1}\bm f^{\mathrm{TE}}_{\ell m}-\bm f^{\mathrm{TE}}_{\ell m}}{\chi}
\end{align}from \eqref{eq:Mie_TElm}, we need to  evaluate a derivative \begin{align}
&\left.\frac{\partial}{\partial\chi}\right|_{\chi=0}\frac{ j_\ell(\sqrt{1+\chi}k|\bm r|)}{ikR\left[j_\ell(\sqrt{1+\chi}kR)\left.\frac{\D[x h_\ell^{(2)}(x)]}{\D x}\right|_{x=kR}- h_\ell^{(2)}(kR)\left.\frac{\D[y j_\ell^{\phantom{(}}(y)]}{\D y}\right|_{y=\sqrt{1+\chi}kR}\right]}
\notag\\={}&\frac{\ell j_\ell(k|\bm r|)-k|\bm r| j_{\ell+1}(k|\bm r|) }{2}-\frac{ikR[\ell j_\ell(kR)-kR j_{\ell+1}(kR)] j_\ell(k|\bm r|)}{2}\left.\frac{\D[x h_\ell^{(2)}(x)]}{\D x}\right|_{x=kR}\notag\\{}&+\frac{ikR[\ell(\ell+1)-k^{2}R^2] j_\ell(kR) j_\ell(k|\bm r|)}{2} h_\ell^{(2)}(kR),\label{eq:BornGexpn}\end{align}where $ \left.\frac{\D[x h_\ell^{(2)}(x)]}{\D x}\right|_{x=kR}=(\ell+1) h_\ell^{(2)}(kR)-kR h_{\ell+1}^{(2)}(kR)$. Reading off the imaginary part, we obtain \begin{align}
\hat\gamma_S\bm f^{\mathrm{TE}}_{\ell m,0}={}&\left\{\frac{kR[\ell j_\ell(kR)-kR j_{\ell+1}(kR)][(\ell+1) j_\ell(kR)-kRj_{\ell+1}^{}(kR)]}{2}\right.\notag\\{}&\left.-\frac{kR[\ell(\ell+1)-k^{2}R^2][ j_\ell(kR)]^{2}}{2}\right\}\bm f^{\mathrm{TE}}_{\ell m,0}\notag\\={}&\frac{k^{2}R^2}{2}\{kR[ j_\ell(kR)]^{2}-(2\ell+1) j_\ell(kR) j_{\ell+1}(kR)+kR[ j_{\ell+1}(kR)]^{2}\}\bm f^{\mathrm{TE}}_{\ell m,0}.\label{eq:gamma_S_TE_lambda}
\end{align}One can convert the eigenvalue above into the integral in \eqref{eq:gammaS_TE}, by Lommel's method \cite[\S5.11(8)]{Watson}. Similarly, we can confirm that \begin{align}&
\hat\gamma_S\bm f^{\mathrm{TM}}_{\ell m}\notag\\={}&\frac{k^{2}R^2}{2}\left\{\left[kR+\frac{2(\ell+1)}{kR}\right][ j_\ell(kR)]^{2}-(2\ell+3) j_\ell(kR) j_{\ell+1}(kR)+kR[ j_{\ell+1}(kR)]^{2}\right\}\bm f^{\mathrm{TM}}_{\ell m},\label{eq:gamma_S_TM_lambda}
\end{align} which is equivalent to \eqref{eq:gammaS_TM}.

As for the   completeness statement, we can check that  the eigenrelation $ \hat\gamma_S\bm f=\lambda\bm f$ for   $\lambda\in\mathbb R\smallsetminus\{0\}$ and $\bm f\in L^2( O(\mathbf0,R);\mathbb C^3),\Vert\bm f\Vert_{L^2( O(\mathbf0,R);\mathbb C^3)}\neq0$
entails the transversality condition $ \nabla\cdot\bm f=0$ and the Helmholtz equation $ (\nabla^2+k^2)\bm f=0$. These two differential equations (both interpreted in the distributional sense) ensure that $\bm f$ has a Fourier--Bessel expansion in terms of a basis spanned by all the  $ \bm f^{\mathrm{TE}}_{\ell m}$ and  $ \bm f^{\mathrm{TM}}_{\ell m}$ wave modes.  \end{proof}

A refinement of the norm bounds for $ \R\hat{\mathscr G}$ will appear in \S\ref{subsec:kRgammaC}. This will draw on the next proposition concerning certain scalar products $\langle \bm F,\bm G\rangle_V \colonequals\iiint_V\bm F^*(\bm r)\cdot\bm G(\bm r)\D^3\bm r$.

\begin{proposition}[Mie projections]\label{prop:Mie_proj} For $V=O(\mathbf0,R)$, and $ \bm E\in L^2(V;\mathbb C^3)$, we have an orthogonal projection to the range of $ \hat\gamma_S$:\begin{align}
\hat P\bm E\colonequals\sum_{\ell=1}^\infty\sum_{m=-\ell}^\ell\left(\frac{\langle\bm f^{\mathrm{TE}}_{\ell m},\bm E\rangle _V\bm f^{\mathrm{TE}}_{\ell m}}{\langle\bm f^{\mathrm{TE}}_{\ell m},\bm f^{\mathrm{TE}}_{\ell m}\rangle _V}+\frac{\langle\bm f^{\mathrm{TM}}_{\ell m},\bm E\rangle _V\bm f^{\mathrm{TM}}_{\ell m}}{\langle\bm f^{\mathrm{TM}}_{\ell m},\bm f^{\mathrm{TM}}_{\ell m}\rangle _V}\right),\label{eq:PranS}
\end{align}where $\bm f^{\mathrm{TE}}_{\ell m}(\bm r)\colonequals
\nabla\times[\bm r j_\ell(k|\bm r|)Y_{\ell m}(\theta,\phi)]$ and $\bm f^{\mathrm{TM}}_{\ell m}(\bm r)\colonequals\nabla\times\bm f^{\mathrm{TE}}_{\ell m,0}(\bm r)$ satisfy\begin{align}
\langle\bm f^{\mathrm{TE}}_{\ell m},\bm f^{\mathrm{TE}}_{\ell 'm'}\rangle _V={}&\delta_{\ell\ell'}\delta_{mm'}\frac{\ell(\ell+1)}{k^{3}}\lambda_{\ell}^{\mathrm {TE}}(kR),\label{eq:flm_TE_fl'm'_TE}\\\langle\bm f^{\mathrm{TM}}_{\ell m},\bm f^{\mathrm{TM}}_{\ell 'm'}\rangle _V={}&\delta_{\ell\ell'}\delta_{mm'}\frac{\ell(\ell+1)}{k}\lambda_{\ell}^{\mathrm {TM}}(kR),\label{eq:flm_TM_fl'm'_TM}\\\langle\bm f^{\mathrm{TE}}_{\ell m},\bm f^{\mathrm{TM}}_{\ell 'm'}\rangle _V={}&0,\label{eq:flm_TE_fl'm'_TM}
\end{align} with $ \delta_{ab}=1$ for $ a=b$ and $ \delta_{ab}=0$ for $a\neq b$.

Furthermore, if we define  $\bm g^{\mathrm{TE}}_{\ell m}(\bm r)\colonequals
\nabla\times[|\bm r|\bm r j_{\ell+1}(k|\bm r|)Y_{\ell m}(\theta,\phi)]$ and $\bm g^{\mathrm{TM}}_{\ell m}(\bm r)\colonequals\nabla\times\bm g^{\mathrm{TE}}_{\ell m}(\bm r)$, then we have  \begin{align}
\langle (\hat I-\hat P)\bm g^{\mathrm{TE}}_{\ell m}, (\hat I-\hat P)\bm g^{\mathrm{TE}}_{\ell' m'}\rangle_V={}&\delta_{\ell\ell'}\delta_{mm'} \Vert  (\hat I-\hat P)\bm g^{\mathrm{TE}}_{\ell m}\Vert_{L^2(V;\mathbb C^3)}^2,\label{eq:glm_TE_gl'm'TE}\\\langle (\hat I-\hat P)\bm g^{\mathrm{TM}}_{\ell m}, (\hat I-\hat P)\bm g^{\mathrm{TM}}_{\ell' m'}\rangle_V={}&\delta_{\ell\ell'}\delta_{mm'} \Vert  (\hat I-\hat P)\bm g^{\mathrm{TM}}_{\ell m}\Vert_{L^2(V;\mathbb C^3)}^{2},\label{eq:glm_TM_gl'm'TM}\\\langle (\hat I-\hat P)\bm g^{\mathrm{TE}}_{\ell m}, (\hat I-\hat P)\bm g^{\mathrm{TM}}_{\ell' m'}\rangle_V={}&0,\label{eq:glm_TE_gl'm'TM}
\end{align}with  bounds\begin{align}
\frac{k^{5}\Vert(\hat I-\hat P)\bm g^{\mathrm{TE}}_{\ell m}\Vert^2_{L^2(V\mathbb C^3)}}{\ell(\ell+1)}\leq{}&(kR)^2\lambda_{\ell+1}^{\mathrm {TE}}(kR),\label{eq:gTE_bd}\\\frac{k^{3}\Vert(\hat I-\hat P)\bm g^{\mathrm{TM}}_{\ell m}\Vert^2_{L^2(V;\mathbb C^3)}}{\ell(\ell+1)}\leq{}&(kR)^2\lambda_{\ell+1}^{\mathrm {TM}}(kR).\label{eq:gTM_bd}
\end{align} \end{proposition}

\begin{proof}First, we observe that the relations $\bm F(\bm r)\colonequals \nabla\times[\bm r f(|\bm r|)Y_{\ell m}^*(\theta,\phi)]=-f(|\bm r|)
\bm r \times\nabla Y_{\ell m}^{*}(\theta,\phi)$,  $\bm G(\bm r)\colonequals\nabla\times[\bm r g(|\bm r|)Y_{\ell' m'}^{}(\theta,\phi)]=-g(|\bm r|)
\bm r \times\nabla Y_{\ell m}^{}(\theta,\phi) $  and\begin{align}
\iiint_V\bm F(\bm r)\cdot\bm G(\bm r)\D^3\bm r={}&\int_0^R f(r)g(r) r^2\D r\varoiint_{|\bm  n|=1}[\bm n\times \nabla Y_{\ell m}^*(\theta,\phi)]\cdot[\bm n\times \nabla Y_{\ell' m'}^{}(\theta,\phi)]\D\Omega\notag\\={}&-\int_0^R f(r)g(r) r^2\D r\varoiint_{|\bm  n|=1} Y_{\ell m}^*(\theta,\phi)\Delta_{S^2}  Y_{\ell' m'}^{}(\theta,\phi)\D\Omega\label{eq:F.G_int1}
\end{align}hold for all smooth functions $ f$ and $g$, where  the Laplace--Beltrami operator $ \Delta_{S^2}\colonequals\frac{1}{\sin\theta}\frac{\partial}{\partial\theta}\left( \sin\theta\frac{\partial}{\partial\theta} \right)+\frac{1}{\sin^2\theta}\frac{\partial^2}{\partial\phi^2}$ satisfies $ \Delta_{S^2} Y_{\ell m}^{}(\theta,\phi)=-\ell(\ell+1)Y_{\ell m}^{}(\theta,\phi)$. Thus, the last displayed formula is equal to $\ell'(\ell'+1) \int_0^R f(r)g(r) r^2\D r\varoiint_{|\bm  n|=1} Y_{\ell m}^*(\theta,\phi)  Y_{\ell' m'}^{}(\theta,\phi)\D\Omega=\delta_{\ell\ell'}\delta_{mm'} \ell'(\ell'+1) \int_0^R f(r)g(r) r^2\D r$. In particular, this implies \eqref{eq:flm_TE_fl'm'_TE}
and consequently $
\frac{\hat \gamma_S\bm f^{\mathrm{TE}}_{\ell m}}{\langle\bm f^{\mathrm{TE}}_{\ell m},\bm f^{\mathrm{TE}}_{\ell m}\rangle _V}=\frac{k^{3}\bm f^{\mathrm{TE}}_{\ell m}}{\ell(\ell+1)}$ [cf.\ \eqref{eq:gammaS_TE}].

Next, we point out that a vector identity $ (\nabla\times\bm F)\cdot(\nabla\times\bm G)=\bm F\cdot(\nabla\times\nabla\times\bm G)+\nabla\cdot[\bm F\times(\nabla\times\bm G)]$  enables us to identify $\iiint_V[\nabla\times\bm F(\bm r)]\cdot[\nabla\times\bm G(\bm r)]\D^3\bm r+\iiint_V\bm F(\bm r)\cdot\nabla^{2}\bm G(\bm r)\D^3\bm r$ with \begin{align}
{}&\varoiint_{\partial V}\bm n\cdot\{\bm F(\bm r)\times[\nabla\times\bm G(\bm r)]\}\D S\notag\\={}&-Rf(R)\varoiint_{\partial V}[\nabla\times\bm G(\bm r)]\cdot\{\bm n\times[\bm n\times\nabla Y_{\ell m}^{*}(\theta,\phi)]\}\D S\notag\\={}&Rf(R)\varoiint_{\partial V}[\nabla\times\bm G(\bm r)]\cdot\nabla Y_{\ell m}^{*}(\theta,\phi)\D S=Rf(R)\varoiint_{\partial V}\nabla\cdot[\bm G(\bm r)\times\nabla Y_{\ell m}^{*}(\theta,\phi)]\D S\label{eq:rotF.rotG_int1}
\end{align}with $ \bm n$ being the outward unit normal vector on $ \partial V$. Here, the last expression can be further simplified into \begin{align}&
Rf(R)\varoiint_{\partial V}\nabla\cdot\left\{ g(|\bm r|)\frac{
\bm r }{|\bm r|^2} [\bm r\times\nabla Y_{\ell m}^{*}(\theta,\phi)][\bm r\times\nabla Y_{\ell' m'}^{}(\theta,\phi)]\right\}\D S\notag\\={}&Rf(R)\varoiint_{\partial V}\nabla\cdot\left[g(|\bm r|)\frac{
\bm r }{|\bm r|^2}\right][\bm r\times\nabla Y_{\ell m}^{*}(\theta,\phi)][\bm r\times\nabla Y_{\ell' m'}^{}(\theta,\phi)]\D S\notag\\={}&Rf(R)\left[ \frac{g(R)}{R^{2}} +\frac{1}{R}\left.\frac{\D g(r)}{\D r}\right|_{r=R}\right]\varoiint_{\partial V}[\bm r\times\nabla Y_{\ell m}^{*}(\theta,\phi)][\bm r\times\nabla Y_{\ell' m'}^{}(\theta,\phi)]\D S,\label{eq:rotF.rotG_int2}
\end{align}which vanishes when $ |\ell-\ell'|+|m-m'|\neq0$, and equals    $ \ell(\ell+1)Rj_\ell(kR)[(\ell+1)j_\ell(kR)-kRj_{\ell+1}(kR)]$ when $ \ell=\ell', m=m',f(r)=g(r)=j_\ell(kr)$. Bearing in mind the Helmholtz equation $(\nabla^{2}+k^{2})\bm  f^{\mathrm{TE}}_{\ell m}= \mathbf0$, we have \eqref{eq:flm_TM_fl'm'_TM} and  $
\frac{\hat \gamma_S\bm f^{\mathrm{TM}}_{\ell m}}{\langle\bm f^{\mathrm{TM}}_{\ell m},\bm f^{\mathrm{TM}}_{\ell m}\rangle _V}=\frac{k\bm f^{\mathrm{TM}}_{\ell m}}{\ell(\ell+1)}$ [cf.\ \eqref{eq:gammaS_TM}, \eqref{eq:gamma_S_TE_lambda} and \eqref{eq:gamma_S_TM_lambda}].

To show \eqref{eq:flm_TE_fl'm'_TM} in a wider context, we consider  $\iiint_V\{\nabla\times[\bm r f(|\bm r|)Y_{\ell m}^*(\theta,\phi)]\}\cdot[\nabla\times\bm G(\bm r)]\D^3\bm r-\iiint_Vf(|\bm r|)Y_{\ell m}^*(\theta,\phi)\bm r \cdot[\nabla\times\nabla\times\bm G(\bm r)]\D^3\bm r$, which is equal to $ \varoiint_{\partial V}\bm n\cdot\{[\bm r f(|\bm r|)Y_{\ell m}^*(\theta,\phi)]\times[\nabla\times\bm G(\bm r)]\}\D S=0$, due to identically vanishing $ \bm n\cdot\{\bm r\times[\nabla\times\bm G(\bm r)]\}$ on the spherical boundary $ \partial V$. When $f(r)=j_\ell(kr)$ and $\bm G(\bm r)=\bm f^{\mathrm{TE}}_{\ell' m'}(\bm r)$, this leaves us an orthogonality relation $ \langle\bm f^{\mathrm{TE}}_{\ell m},\bm f^{\mathrm{TM}}_{\ell 'm'}\rangle _V=0$ as stated in  \eqref{eq:flm_TE_fl'm'_TM}, because $\bm r \cdot[\nabla\times\nabla\times\bm G(\bm r)]=-\bm r \cdot\nabla^{2}\bm G(\bm r)=k^{2}\bm r \cdot\bm f^{\mathrm{TE}}_{\ell' m'}(\bm r)$ is also identically vanishing.
Thus far, we can already transcribe   $\hat \gamma_S\bm F=\hat \gamma_S\hat P\bm F$ (a consequence of Proposition \ref{prop:Mie_gamma_S}) into \begin{align}
\hat\gamma_S\bm F=\sum_{\ell=1}^\infty\sum_{m=-\ell}^\ell\frac{k\left(k^{2}\langle\bm f^{\mathrm{TE}}_{\ell m},\bm F\rangle _V\bm f^{\mathrm{TE}}_{\ell m}+\langle\bm f^{\mathrm{TM}}_{\ell m},\bm F\rangle _V \bm f^{\mathrm{TM}}_{\ell m}\right)}{\ell(\ell+1)}\label{eq:gamma_S_flm}
\end{align}for each fixed $ \bm F\in L^2(V;\mathbb C^3)$. This spectral decomposition  will be useful later in \S\ref{subsec:kRgammaC}.

From the foregoing arguments surrounding \eqref{eq:F.G_int1} and our proof of  \eqref{eq:flm_TE_fl'm'_TM}, we already know that $ \langle \bm g^{\mathrm{TE}}_{\ell m},\bm f^{\mathrm{TE}}_{\ell 'm'}\rangle _V=\delta_{\ell\ell'}\delta_{mm'}\langle \bm g^{\mathrm{TE}}_{\ell m},\bm f^{\mathrm{TE}}_{\ell m}\rangle _V $ and  $ \langle \bm g^{\mathrm{TE}}_{\ell m},\bm f^{\mathrm{TM}}_{\ell 'm'}\rangle _V=0 $, so $ \hat P  \bm g^{\mathrm{TE}}_{\ell m}=\frac{\langle \bm g^{\mathrm{TE}}_{\ell m},\bm f^{\mathrm{TE}}_{\ell m}\rangle _V}{\langle\bm f^{\mathrm{TE}}_{\ell m},\bm f^{\mathrm{TE}}_{\ell m}\rangle _V}\bm f^{\mathrm{TE}}_{\ell m}$. Meanwhile, it is also clear from \eqref{eq:rotF.rotG_int1}--\eqref{eq:rotF.rotG_int2} that $ \langle \bm g^{\mathrm{TM}}_{\ell m},\bm f^{\mathrm{TM}}_{\ell 'm'}\rangle _V=\delta_{\ell\ell'}\delta_{mm'}\langle \bm g^{\mathrm{TM}}_{\ell m},\bm f^{\mathrm{TM}}_{\ell m}\rangle _V$. To demonstrate that   $ \langle\bm f^{\mathrm{TE}}_{\ell m}, \bm g^{\mathrm{TM}}_{\ell' m'}\rangle _V=0$, we reexamine our verification of  $ \langle\bm f^{\mathrm{TE}}_{\ell m},\bm f^{\mathrm{TM}}_{\ell 'm'}\rangle _V=0$, trading $ \bm G(\bm r)$ for $ \bm g^{\mathrm{TE}}_{\ell' m'}(\bm r)$, and noting that $ \nabla\times\nabla\times\bm g^{\mathrm{TE}}_{\ell' m'}(\bm r)=-\nabla^{2}\bm g^{\mathrm{TE}}_{\ell' m'}(\bm r)=\bm r\times
\nabla\{\nabla^{2}[|\bm r| j_{\ell+1}(k|\bm r|)Y_{\ell' m'}(\theta,\phi)]\}$ is pointwise orthogonal to $\bm r$. Therefore, it is also true that  $ \hat P  \bm g^{\mathrm{TM}}_{\ell m}=\frac{\langle \bm g^{\mathrm{TM}}_{\ell m},\bm f^{\mathrm{TM}}_{\ell m}\rangle _V}{\langle\bm f^{\mathrm{TM}}_{\ell m},\bm f^{\mathrm{TM}}_{\ell m}\rangle _V}\bm f^{\mathrm{TM}}_{\ell m}$.

The computations in the last paragraph can be adapted to a confirmation of the identities  $ \langle \bm g^{\mathrm{TE}}_{\ell m},\bm g^{\mathrm{TE}}_{\ell 'm'}\rangle _V=\delta_{\ell\ell'}\delta_{mm'}\langle \bm g^{\mathrm{TE}}_{\ell m},\bm g^{\mathrm{TE}}_{\ell m}\rangle _V $,  $ \langle \bm g^{\mathrm{TE}}_{\ell m},\bm g^{\mathrm{TM}}_{\ell 'm'}\rangle _V=0 $, and $\langle\bm g^{\mathrm{TE}}_{\ell m}, \bm g^{\mathrm{TM}}_{\ell' m'}\rangle _V=0 $. To show that   $ \langle \bm g^{\mathrm{TM}}_{\ell m},\bm g^{\mathrm{TM}}_{\ell 'm'}\rangle _V=\delta_{\ell\ell'}\delta_{mm'}\langle \bm g^{\mathrm{TM}}_{\ell m},\bm g^{\mathrm{TM}}_{\ell m}\rangle _V $, we need an additional observation that $\nabla^{2}\bm g^{\mathrm{TE}}_{\ell' m'}(\bm r) =\bm r\times
\nabla\{[k^{2}|\bm r| j_{\ell+1}(k|\bm r|)-2k j_{\ell}(k|\bm r|)]Y_{\ell' m'}(\theta,\phi)\}$    is orthogonal to $ \bm g^{\mathrm{TM}}_{\ell m}(\bm r)$ when $ |\ell-\ell'|+|m-m'|\neq0$.

Combining the results from the last two paragraphs, we arrive at \eqref{eq:glm_TE_gl'm'TE}--\eqref{eq:glm_TE_gl'm'TM}.

The inequality  \eqref{eq:gTE_bd} paraphrases the facts that $ \Vert(\hat I-\hat P)\bm g^{\mathrm{TE}}_{\ell m}\Vert^2_{L^2(V;\mathbb C^3)}\leq\Vert\bm g^{\mathrm{TE}}_{\ell m}\Vert^2_{L^2(V;\mathbb C^3)}$ and $ \int_0^{kR}j_{\ell+1}^{2}(\xi)\xi^4\D \xi\leq(kR)^2\int_0^{kR}j_{\ell+1}^{2}(\xi)\xi^2\D \xi=(kR)^2\lambda_{\ell+1}^{\mathrm {TE}}(kR)$. We may enlist the help from  \eqref{eq:rotF.rotG_int1}--\eqref{eq:rotF.rotG_int2} to compute  $\Vert\bm g^{\mathrm{TM}}_{\ell m}\Vert^2_{L^2(V;\mathbb C^3)} $, as follows:\begin{align}&
\Vert\bm g^{\mathrm{TM}}_{\ell m}\Vert^2_{L^2(V;\mathbb C^3)} -k^2\Vert\bm g^{\mathrm{TE}}_{\ell m}\Vert^2_{L^2(V;\mathbb C^3)}+\frac{2\ell(\ell+1)}{k^{3}}\int_0^{kR}j_{\ell}(\xi)j_{\ell+1}(\xi)\xi^3\D\xi\notag\\={}&\ell(\ell+1)R^{4}j_{\ell+1}(kR)\left[ \frac{2j_{\ell+1}(kR)}{R^{}} +\left.\frac{\partial j_{\ell+1}(kr)}{\partial r}\right|_{r=R}\right].
\end{align} Thus,  by equating $\frac{k^{3}\Vert\bm g^{\mathrm{TM}}_{\ell m}\Vert^2_{L^2(V;\mathbb C^3)}}{\ell(\ell+1)}-(kR)^2\lambda_{\ell+1}^{\mathrm {TM}}(kR)$ with \begin{align}{}&\int_0^{kR}j_{\ell+1}^{2}(\xi)\xi^4\D \xi-2\int_0^{kR}j_{\ell}(\xi)j_{\ell+1}(\xi)\xi^3\D\xi\notag\\{}&+(kR)^{3}j_{\ell+1}(kR)[kRj_{\ell}(kR)-\ell j_{\ell+1}(kR)]-(kR)^2\lambda_{\ell+1}^{\mathrm {TM}}(kR)\notag\\\leq{}&(kR)^2[\lambda_{\ell+1}^{\mathrm {TE}}(kR)-\lambda_{\ell+1}^{\mathrm {TM}}(kR)]-2\int_0^{kR}j_{\ell}(\xi)j_{\ell+1}(\xi)\xi^3\D\xi\notag\\{}&+(kR)^{3}j_{\ell+1}(kR)[kRj_{\ell}(kR)-\ell j_{\ell+1}(kR)]=-(2\ell+1)\lambda_{\ell+1}^{\mathrm {TE}}(kR)\leq0,\end{align}we can prove   \eqref{eq:gTM_bd}.   \end{proof}
\section{Norm bounds for the Green operator $ \hat{\mathscr G}$\label{sec:norm_bds}}
\subsection{$L^2$-norm bounds of $ \Vert\hat\gamma\Vert$\label{subsec:gamma_L2_bd}}In the quantitative analysis of the Born operator, we need to bound the operator norm $\Vert\hat\gamma\Vert_{L^2(V;\mathbb C^3)}  $ for
\begin{align}(\hat\gamma\bm E)(\bm r)\colonequals \iiint_V\hat {\Gamma}(\bm r,\bm r')\bm E(\bm r')\D^3 \bm r'\end{align}
where
the  $3\times3$  matrix $\hat{\Gamma}(\bm r,\bm r')$ is given by [cf.~\eqref{eq:Hess_prop}]\begin{align}\label{eq:Gamma_mat}\hat{\Gamma}(\bm r,\bm r')={}&\frac{k^2 }{4\pi|\bm r-\bm r' |}\left\{\left[e^{-ik|\bm r-\bm r' |}+ \frac{1-e^{-ik|\bm r-\bm r' |}\left(1+ik|\bm r-\bm r' |\right)}{k^2|\bm r-\bm r' |^2} \right]\mathbf1+\right.\notag\\&\left.+\left[ 3\frac{e^{-ik|\bm r-\bm r' |}\left(1+ik|\bm r-\bm r' |\right)-1}{k^2|\bm r-\bm r' |^2}-e^{-ik|\bm r-\bm r' |} \right]\hat {\bm r}\hat{\bm r}^\mathrm{T}\right\}.\end{align}Here,
the  $3\times3$  matrix\begin{align}\hat {\bm r}\hat{\bm r}^\mathrm{T}=\frac{\bm r-\bm r' }{|\bm r-\bm r' |}\left( \frac{\bm r-\bm r' }{|\bm r-\bm r' |} \right)^\mathrm{T}\end{align}satisfies $ \hat {\bm r}\hat{\bm r}^{\mathrm T}\hat {\bm r}=\hat {\bm r}$, so it has an eigenvalue $\lambda_{1}=1$, with eigenvector $\hat {\bm r} $. The three eigenvalues of $ \hat {\bm r}\hat{\bm r}^\mathrm{T}$ are actually ${\lambda}_1=1,\lambda_2=\lambda_3=0$,
as evident from the computations $\Tr(\hat {\bm r}\hat{\bm r}^{\mathrm T})=\lambda_{1}+\lambda_{2}+\lambda_{3}=1 $ and  $\det(\hat {\bm r}\hat{\bm r}^{\mathrm T})=\lambda_{1}\lambda_{2}\lambda_{3}=0.$ (Hereafter, without affecting the evaluation of the integral, we may always  assume that $ \bm r\neq\bm r' $, so that  the matrix is never ill-defined.)

We may accordingly define $\hat{\Gamma}_C (\bm r,\bm r')=\R\hat{\Gamma}(\bm r,\bm r')$ and $\hat{\Gamma}_S (\bm r,\bm r')=-\I\hat{\Gamma}(\bm r,\bm r') $, so that $(\hat\gamma_{C}\bm E)(\bm r)=\iiint_V\hat {\Gamma}_C(\bm r,\bm r')\bm E(\bm r')\D^3 \bm r' $ and  $(\hat\gamma_{S}\bm E)(\bm r)=\iiint_V\hat {\Gamma}_S(\bm r,\bm r')\bm E(\bm r')\D^3 \bm r' $. The bound estimate for the operator norm of $\hat {\gamma}=\hat\gamma_C-i\hat\gamma_S $ can be then obtained from $\Vert\hat \gamma\Vert_{L^2(V;\mathbb C^3)}\leq\Vert\hat \gamma_C\Vert_{L^2(V;\mathbb C^3)}+\Vert\hat \gamma_S\Vert_{L^2(V;\mathbb C^3)} $.

\begin{proposition}\label{prop:AreaBound}
We have the following upper bound estimate  \begin{align}\max\{\Vert\hat \gamma_C\Vert_{L^2(V;\mathbb C^3)},\Vert\hat \gamma_S\Vert_{L^2(V;\mathbb C^3)}\}\leq{}&     \frac{3 }{5\pi}\left(k^3\iiint_V\D^3\bm r\right)^{2/3}. \end{align} \end{proposition}\begin{proof}For brevity, we will work out  the proof for $\hat{\gamma}_C$ only, as the same procedure, \textit{mutatis mutandis}, will also apply {to $\hat {\gamma}_S$}.

It is easy to recognize that  $\hat \gamma_C $ is a  compact Hermitian operator. The corresponding absolute value of the Rayleigh quotient $|\langle\bm E,\hat{\gamma}_C\bm E\rangle_{V} /\langle\bm E,\bm E\rangle_{V}| $ (defined for $\Vert\bm E\Vert_{L^2(V;\mathbb C^3)}\neq0$) attains its maximum  when  $\bm E$ is an eigenvector  subordinate to the eigenvalue of $\hat \gamma_C $  with the largest modulus, and we  have the equality $\Vert\hat \gamma_C\Vert_{L^2(V;\mathbb C^3)}= \max_{\Vert\bm E\Vert_{L^2(V;\mathbb C^3)}\neq0}|\langle\bm E,\hat{\gamma}_C\bm E\rangle_{V} /\langle\bm E,\bm E\rangle_{V}|=\max_{\Vert\bm E\Vert_{L^2(V;\mathbb C^3)}=1}|\langle\bm E,\hat{\gamma}_C\bm E\rangle_{V}|$ as well.
Consequently, we just need to bound the absolute value of the quadratic form $\langle\bm E,\hat{\gamma}_C\bm E\rangle_{V}$ from above for an upper bound estimate of  $\Vert\hat \gamma_C\Vert_{L^2(V;\mathbb C^3)} $.

Now, we have the inequality\footnote{Hereafter, an asterisk in the superscript denotes complex conjugation of numbers and vectors, and conjugate transpose of matrices.}
\begin{align}|\langle\bm E,\hat{\gamma}_C\bm E\rangle_{V}|={}&\left\vert\iiint_V\iiint_V\bm E^{*}(\bm r)\cdot[\hat {\Gamma}_C(\bm r,\bm r')\bm E(\bm r')]\D^3 \bm r'\D^3\bm r\right\vert\notag\\\leq{}&\iiint_V\iiint_V|\bm E^{*}(\bm r)||\hat {\Gamma}_C(\bm r,\bm r')\bm E(\bm r')|\D^3 \bm r'\D^3\bm r\notag\\\leq{}&\iiint_V\iiint_V|\bm E^{*}(\bm r)|\left\{|\lambda|_{\max}\left[\hat {\Gamma}_C(\bm r,\bm r')\right]\right\}|\bm E(\bm r')|\D^3 \bm r'\D^3\bm r\end{align}where $|\lambda|_{\max} $ extracts the eigenvalue  with the largest modulus.  As the  $3\times3$  symmetric matrix  $ \hat{\Gamma}_C(\bm r,\bm r')=a\mathbf1+b\hat{\bm r}\hat{\bm r}^\mathrm{T}$  has three eigenvalues $\lambda _{1}=a+b,\lambda_{2}=\lambda_3=a$, we may deduce that
\begin{align}&|\lambda|_{\max}\left[\hat {\Gamma}_C(\bm r,\bm r')\right]\notag\\={}&\frac{k^2 }{4\pi|\bm r-\bm r' |}\max\left\{ \left\vert\R \left(e^{-i\xi}+ \frac{1-e^{-i\xi}\left(1+i\xi\right)}{\xi^2}\right) \right\vert,2\left\vert\R \left( \frac{1-e^{-i\xi}\left(1+i\xi\right)}{\xi^2}\right)\right\vert  \right\}, \end{align}where $ \xi=k|\bm r-\bm r' |$.

 We claim that the  formula above brings us \begin{align}|\lambda|_{\max}\left[\hat {\Gamma}_C(\bm r,\bm r')\right]\leq\frac{k^2 }{4\pi|\bm r-\bm r' |}\sqrt{\frac{13}{12}-\frac{3}{10 \pi ^2}}.\label{eq:l_max}\end{align}

For $\xi>0 $, we have $|e^{i\xi}-1-i\xi|=|\int_0^{\xi}(e^{i\eta}-1)\D \eta|\leq\int_0^{\xi}|e^{i\eta}-1|\D \eta=\int_0^{\xi}|2\sin\frac{\eta}{2}|\D \eta\leq \frac{\xi^2}{2} $,  thus
$2\left|1-e^{-i\xi}\left(1+i\xi\right)\right|\leq{\xi^2}$. The function $ \Lambda(\xi)\colonequals\frac{1}{\xi^{2}}|e^{-i\xi}\xi^{2}+1-e^{-i\xi}\left(1+i\xi\right)|$ satisfies $ \Lambda(\xi)\leq\sqrt{[\Lambda(\xi)]^2+\frac{\xi^{4}}{48}\{[b'(\xi)]^{2}-b(\xi)b''(\xi)\}}=\sqrt{\frac{13}{12}-\frac{10 }{3 \xi ^2}\sin ^2\frac{\xi }{2}}$ for all $ \xi>0$, where $ b(\xi)\colonequals\frac{2}{\xi}j_1\big(\frac{\xi }2\big)=\frac{2}{\xi^{2}}\big(\frac{4 }{\xi }\sin\frac{\xi }{2}-2 \cos\frac{\xi }{2}\big)$ is a member of  the  Laguerre--P\'olya class \cite{Jensen1913,Patrick1973}.
In particular, this implies\begin{align}\Lambda(\xi)\leq{}&\sqrt{\frac{13}{12}-\frac{10}{3(\frac{5\pi }{3}) ^2} \sin ^2\frac{5\pi }{6}}=\sqrt{\frac{13}{12}-\frac{3}{10 \pi ^2}}=1.0261+\end{align}when $ 0<\xi\leq\frac{5\pi }{3}$. Meanwhile, we have
\begin{align}\Lambda(\xi)={}&\sqrt{\left(1-\frac{1-\cos\xi}{\xi^{2}}\right)^{2}+\frac{(\xi-\sin\xi)^{2}}{\xi^{4}}}\leq\sqrt{1+\left(\frac{\xi+1}{\xi^{2}}\right)^2}\leq\sqrt{1+\left[\frac{\frac{5\pi}{3}+1}{\left(\frac{5\pi}{3}\right)^{2}}\right]^2}=1.0255+\end{align}when $ \xi\geq\frac{5\pi }{3}$. This establishes \eqref{eq:l_max} as claimed. [One may  replace the constant $\sqrt{\frac{13}{12}-\frac{3}{10 \pi ^2}}=1.0261+ $ by the numerical value of $\max_{\xi>0}\Lambda(\xi)=1.0140+ $, but the improvement is not significant.]

Using the  Hardy--Littlewood--Sobolev inequality \cite{HL1,HL2,Sob} with Lieb's sharp constant (see \cite[Theorem 3.1]{Lieb1983}, as well as \cite[Theorem~4.3]{Lieb})
\begin{align}&\left\vert\int_{\mathbb R^d}\int_{\mathbb R^d}\frac{f(\bm r')h(\bm r)}{|\bm r-\bm r'|^m}\D^d\bm r'\D^d\bm r\right\vert\notag\\\leq{}& \pi^{m/2}\frac{\Gamma(\frac{d-m}{2})}{\Gamma(d-\frac{m}{2})}\left[ \frac{\Gamma(\frac{d}{2})}{\Gamma(d)} \right]^{-1+\frac{m}{d}}\Vert f\Vert_{L^{2d/(2d-m)}(\mathbb R^d;\mathbb C)}\Vert h\Vert_{L^{2d/(2d-m)}(\mathbb R^d;\mathbb C)},\notag\\&\text{where }d=3,m=1,\frac{2d}{2d-m}=\frac{6}{5};\quad f(\bm r)=h(\bm r)=\begin{cases}|\bm E(\bm r)|, & \bm r\in V \\
0, & \bm r\notin V \\
\end{cases},\end{align}
we may deduce from \eqref{eq:l_max} the following estimate
\begin{align}|\langle\bm E,\hat{\gamma}_C\bm E\rangle_{V}|\leq{}&\frac{2k^2 }{3\pi}\sqrt[3]{\frac{2}{\pi}}\sqrt{\frac{13}{12}-\frac{3}{10 \pi ^2}}\Vert \bm E\Vert_{L^{6/5}(\mathbb R^d;\mathbb C^3)}^{2}\notag\\\leq{}&\frac{2\langle\bm E,\bm E\rangle_{V} }{3\pi}\sqrt[3]{\frac{2}{\pi}}\sqrt{\frac{13}{12}-\frac{3}{10 \pi ^2}}\left(k^3\iiint_V\D^3\bm r\right)^{2/3}\leq\frac{3 }{5\pi}\left(k^3\iiint_V\D^3\bm r\right)^{2/3}.\end{align} Here, the H\"older inequality \cite[Theorem 2.3]{Lieb} has been used to derive the estimate $\Vert \bm E\Vert_{L^{6/5}(V;\mathbb C^3)}$ $\leq\Vert \bm E\Vert_{L^{2}(V;\mathbb C^3)}(\iiint_V\D^3\bm r)^{1/3} $, and the ratio $\frac23(\frac{2}{\pi})^{1/3}\sqrt{\frac{13}{12}-\frac{3}{10 \pi ^2}}:\frac35 $ lies between $98\%$ and  $99\%$.
 \end{proof} \begin{remark}When $V=\{(x,y,z)\in\mathbb R^3|x^2+y^2<R^2,0<z<h\} $ takes the shape of a cylindrical ``antenna'', we may compute that $\max\{\Vert\hat \gamma_C\Vert_{L^2(V;\mathbb C^3)},\Vert\hat \gamma_S\Vert_{L^2(V;\mathbb C^3)}\}\leq \frac25(\frac12+\frac\pi6)(k^3 R^2 h)^{2/3}$, and the upper bound $\Vert\hat \gamma\Vert_{L^2(V;\mathbb C^3)}\leq\frac45(\frac12+\frac\pi6)(k^3 R^2 h)^{2/3} $ scales as $(kh)^{2/3} $ for large heights $h\gg2\pi/k$. When $V=O(\mathbf0,R)$ is an open ball with radius $R$, the reasoning above leads to an estimate $\max\{\Vert\hat \gamma_C\Vert_{L^2(O(\mathbf0,R);\mathbb C^3)},\Vert\hat \gamma_S\Vert_{L^2(O(\mathbf0,R);\mathbb C^3)}\}$ $  \leq\frac12(kR)^2$. Thus, we may also conclude that $\Vert\hat \gamma\Vert_{L^2(O(\mathbf0,R);\mathbb C^3)}$ $\leq (kR)^2 $. Lacking a treatment of oscillatory cancelations in the convolution kernels, the quadratic growth rate $(kR)^2$  overestimates the operator norm for spheres of large radii  $R\gg2\pi/k$, as will be revealed by the analysis in \S\S\ref{subsec:totalsc}--\ref{subsec:kRgammaC}  below.

\end{remark}
\subsection{$L^2$-norm bounds of  $ \Vert\I\hat{\mathscr G}\Vert$ and total scattering cross-section \label{subsec:totalsc}}In the $ L^2$-formulation of electromagnetic scattering, there is an energy conservation law (generalized optical theorem) for the Born operator $\hat{\mathscr B}=\hat I-\chi\hat {\mathscr G}\colon L^2(V;\mathbb C^3)\longrightarrow L^2(V;\mathbb C^3) $ as given below: \begin{align}&
\frac{|\chi|^2 k^5}{16\pi^2 } \varoiint_{|\bm n|=1}\left|\bm n\times\iiint_V \bm E(\bm r' ) e^{ik\bm n \cdot \bm r' }  \D^3\bm r' \right|^2\D\Omega\notag\\=&\I\left[\chi k^2\iiint_V|\bm E(\bm r )|^{2}\D^3\bm r \right]-\I\left[\chi k^2\iiint_V(\hat{\mathscr  B}\bm E)^*(\bm r)\cdot\bm E(\bm r )\D^3\bm r\right],\label{eq:GOT1}
\end{align}where $ \D\Omega=\sin\theta\D\theta\D\phi$ stands for the infinitesimal steric angles. [Physically speaking, the generalized optical theorem \eqref{eq:GOT1} represents the  energy conservation during the scattering process, and can be rewritten as \begin{align}\label{eq:GOT01}
\sigma_\mathrm{sc}+\sigma_\text{abs}=\frac{4\pi}{k} \I \left[-\frac{\chi k^2}{4\pi} \iiint_V \bm E_\mathrm{inc}^* (\bm r)\cdot \bm E(\bm r) \D^3\bm r \right].
\end{align}
Here, the total scattering cross-section $\sigma_\mathrm{sc}$ (which may be regarded as electromagnetic energy that is radiated, in unit time, to infinite distances) is given by
\begin{align}
\sigma_\mathrm{sc}\colonequals\frac{|\chi|^2 k^4}{16\pi^2 } \varoiint_{|\bm n|=1}\left|\bm n\times\iiint_V \bm E(\bm r' ) e^{ik\bm n \cdot \bm r' }  \D^3\bm r' \right|^2\D\Omega\geq0;
\end{align}
while the total absorption cross-section $\sigma_\text{abs}$ (which may be phenomenologically interpreted as dissipation of electromagnetic energy in unit time due to Joule heating) is given by\begin{align}
\sigma_\text{abs}\colonequals\I\left[-\chi k\iiint_V\bm E^*(\bm r)\cdot \bm E(\bm r)  \D^3\bm r \right]\geq0.
\end{align}The radiated energy ($ \sigma_{\text{sc}}$) and the dissipated energy ($ \sigma_\text{abs}$) are dispatched from the electromagnetic work done by the incident field, which is the right-hand side of \eqref{eq:GOT01}. If there is no absorption, and the incident field is a plane wave  $ \bm E_\mathrm{inc} (\bm r)\propto\bm e_ze^{-ik\bm e_x\cdot\bm r}$,  then \eqref{eq:GOT01} equates the total scattering section on the left-hand side with the imaginary part of forward scattering amplitude on the right-hand side, which is the statement of the standard optical theorem \cite[\S10.11]{Jackson:EM}. This explains the name ``generalized optical theorem''.]

 From an equivalent form of  \eqref{eq:GOT1}:\footnote{Henceforth, we define the Fourier  transform of an electric field $ \bm E\in L^2(V;\mathbb C^3)$ via $ \widetilde{\bm E}(\bm q)\colonequals\iiint_V\bm E(\bm r)e^{i\bm q\cdot\bm r}\D^3\bm r$.}\begin{align}&\iiint_V\left[\nabla\times \nabla\times \iiint_V\frac{\bm E^*(\bm r' )\sin(k|\bm r-\bm r'|)}{4\pi|\bm r-\bm r' |}  \D^3 \bm r'\right]\cdot\bm E(\bm r)\D^3\bm r\notag\\={}&\I\langle\hat{\mathscr  G}\bm E,\bm E\rangle_{V}=\frac{k^3}{16\pi^2 }\varoiint_{|\bm n|=1}\left\vert\bm n\times\widetilde{\bm E}(k\bm n)\right\vert^2\D\Omega,\end{align} we know that the total scattering cross-section (up to a scaling factor) is determined by \begin{align}\langle\hat{\gamma}_S\bm E,\bm E\rangle_{V}={}&\iiint_V\left[\nabla\times \nabla\times \iiint_V\frac{\bm E^*(\bm r' )\sin(k|\bm r-\bm r'|)}{4\pi|\bm r-\bm r' |}  \D^3 \bm r'\right]\cdot\bm E(\bm r)\D^3\bm r\notag\\={}&\frac{k^3}{16\pi^2 }\varoiint_{|\bm n|=1}\left\vert\bm n\times\widetilde{\bm E}(k\bm n)\right\vert^2\D\Omega,\label{eq:gammaS}\end{align}where the linear mapping $\hat \gamma_S=-\I\hat{\mathscr G}\colon {L^2(V;\mathbb C^3)}\longrightarrow{L^2(V;\mathbb C^3)} $ [defined in \eqref{eq:gamma_S}] is a positive-semidefinite and compact operator. The operator norm $\Vert\hat\gamma_S\Vert_{L^2(V;\mathbb C^3)} $ is thus equal to the largest eigenvalue of  the linear mapping $\hat \gamma_S\colon {L^2(V;\mathbb C^3)}\longrightarrow{L^2(V;\mathbb C^3)}$.
\begin{lemma}\label{lm:V_mono}
If $V_{1}\subset V_{2} $, then  $\Vert\hat\gamma_S\Vert_{L^2(V_{1};\mathbb C^3)}\leq \Vert\hat\gamma_S\Vert_{L^2(V_{2};\mathbb C^3)} $.
\end{lemma}
\begin{proof}Suppose that $\bm E_{V_1}\in L^2(V_{1};\mathbb C^3)  $ solves the eigenequation $\hat \gamma_S\bm E_{V_1}=\Vert\hat\gamma_S\Vert_{L^2(V_{1};\mathbb C^3)} \bm E_{V_1}$, and satisfies $\Vert\bm E_{V_1}\Vert_{L^2(V_1;\mathbb C^3)}=1 $. Then, we may take \begin{align}\bm E(\bm r)=\begin{cases}\bm E_{V_1}(\bm r), & \bm r\in V_1 \\
\mathbf0, & \bm r\in V_2\smallsetminus V_1 \\
\end{cases} \end{align}so that $\Vert\bm E\Vert_{L^2(V_2;\mathbb C^3)}=1$ and \begin{align}&\Vert\hat\gamma_S\Vert_{L^2(V_{1};\mathbb C^3)} =\langle\hat{\gamma}_S\bm E_{V_1},\bm E_{V_1}\rangle_{V_1}=\langle\hat{\gamma}_S\bm E,\bm E\rangle_{V_2}\notag\\\leq{}&\Vert\hat\gamma_S\Vert_{L^2(V_{2};\mathbb C^3)} \Vert\bm E\Vert_{L^2(V_2;\mathbb C^3)}^2=\Vert\hat\gamma_S\Vert_{L^2(V_{2};\mathbb C^3)} \end{align}as claimed.
 \end{proof}\begin{remark}\label{rmk:V_mono}In a similar vein, one can argue that if $V_{1}\subset V_{2} $, then  $\Vert\R\hat{\mathscr G}\Vert_{L^2(V_{1};\mathbb C^3)}\leq \Vert\R\hat{\mathscr G}\Vert_{L^2(V_{2};\mathbb C^3)} $. Here, for each $ \varepsilon>0$, one can fine a certain  $\bm E_{V_1,\varepsilon}\in L^2(V_{1};\mathbb C^3)  $   such that  $\Vert\bm E_{V_1,\varepsilon}\Vert_{L^2(V_1;\mathbb C^3)}=1 $ and $ \Vert\R\hat{\mathscr G}\Vert_{L^2(V_{1};\mathbb C^3)}-\varepsilon\leq \Vert(\R\hat{\mathscr G})\bm E_{V_1,\varepsilon}\Vert_{L^2(V_{1};\mathbb C^3)}$. Zero-padding as in the proof above, we  have $ \Vert\R\hat{\mathscr G}\Vert_{L^2(V_{1};\mathbb C^3)}-\varepsilon\leq \Vert\R\hat{\mathscr G}\Vert_{L^2(V_{2};\mathbb C^3)}$ for arbitrary $ \varepsilon>0$.
 \eor\end{remark}
\begin{proposition}For any bounded and open volume $V$, we have \begin{align}\Vert\hat\gamma_S\Vert_{L^2(V;\mathbb C^3)}\leq{}
\sup_{\ell\in\mathbb Z_{>0}}\max\left\{\int_0^{kR_V}j_\ell^2(\xi)\xi^2\D \xi,\frac{}{}\int_0^{kR_{V}}\frac{(\ell+1)j_{\ell-1}^2(\xi)+\ell j_{\ell+1}^2(\xi)}{2\ell+1}\xi^2\D \xi \right\},\label{eq:gamma_S_bd}\end{align}where   $R_V\colonequals \min_{\bm r\in\mathbb R^3}\max_{\bm r'\in V\cup\partial V}|\bm r'-\bm r|$ is the minimal radius of all the circumscribed spheres, and $j_\ell(x)\colonequals (-x)^{\ell}(\frac{1}{x}\frac{\D}{\D x})^\ell(\frac{\sin x}{x}) $ denotes the $\ell$-th order spherical Bessel function.
\end{proposition}
\begin{proof}This is a straightforward consequence of the complete set of eigenvectors for $ \hat\gamma_S$ in Mie scattering (Proposition \ref{prop:Mie_gamma_S}) and the monotonicity of operator norm with respect to set inclusions (Lemma \ref{lm:V_mono}).
 \end{proof}

\begin{proposition}\label{prop:kRVgammaS}We have the following inequality
\begin{align} \Vert\hat\gamma_S\Vert_{L^2(V;\mathbb C^3)}\leq{}&\min\left\{ \frac{k^3}{4\pi}\iiint_V\D^3\bm r ,\frac{11}{45}( kR_V)^3,\frac{\sqrt{2}}{3}( kR_V)^2,\frac{3}{4}(kR_V)^{4/3}, \frac{9}{16}kR_V\right\}\notag\\={}&\min\left\{ \frac{k^3}{4\pi}\iiint_V\D^3\bm r ,\frac{11}{45}( kR_V)^3, \frac{9}{16}kR_V\right\}.\label{eq:gammaS_bd}\end{align}
\end{proposition}\begin{proof}First, we point out that the identity  $\sum_{L=0}^\infty(2L+1) j_L^2(\xi)=1$ (see \cite[(10.1.50)]{Stegun} or \cite[\S5.51(2)]{Watson}) entails $  j_\ell^2(\xi)\leq\frac1{2\ell+1}$, thus resulting in $\int_0^{kR_V}j_\ell^2(\xi)\xi^2\D \xi\leq\frac1{2\ell+1}\int_0^{kR_V}\xi^2\D \xi\leq \frac{( kR_V)^3}{3(2\ell+1)}$ for all $\ell\in\mathbb Z_{\geq 0}$. Hence, \begin{align}\frac{
\Vert\hat\gamma_S\Vert_{L^2(O(\mathbf0,R_V);\mathbb C^3)}}{( kR_V)^3}\leq\sup_{\ell\in\mathbb Z_{>0}}\max\left\{\frac{1}{3(2\ell+1)},\frac{1}{3(2\ell+1)}\left(\frac{\ell+1}{2\ell-1}+\frac{\ell}{2\ell+3}\right)\right\}=\frac{11}{45}.
\end{align}For arbitrarily shaped $V$,  we have $4\pi\Vert\hat\gamma_S\Vert_{L^2(V;\mathbb C^3)}\leq k^3\iiint_V\D^3\bm r $, after we plug the Cauchy--Schwarz inequality $ |\widetilde{\bm E}(\bm q)|\leq\Vert\bm E\Vert_{L^2(V;\mathbb C^3)}\sqrt{\iiint_V\D^3\bm r}$ into the right-hand side of \eqref{eq:gammaS}. This general inequality entails a weaker bound for spheres, which reads $ \Vert\hat\gamma_S\Vert_{L^2(O(\mathbf0,R_V);\mathbb C^3)}\leq\frac13( kR_V)^3$.

Next, we note that
the  following identity (see \cite[(10.1.52)]{Stegun} or \cite[\S5.51(6)]{Watson}) \begin{align}\sum_{L=0}^\infty j_L^2(\xi)=\frac{1}{2\xi}\int_0^{2\xi}\frac{\sin\eta}{\eta}\D\eta,\end{align}along with the inequality $\int_0^{2\xi}\frac{\sin\eta}{\eta}\D\eta\leq\int_0^{\pi}\frac{\sin\eta}{\eta}\D\eta\leq\int_0^\pi\cos\frac{\eta}{2}\cos\frac{\eta}{4}\D\eta=\frac{4\sqrt{2}}3$ for $\xi>0 $ lead us to
$\int_0^{kR_V}j_\ell^2(\xi)\xi^2\D \xi\leq\frac{2\sqrt{2}}3\int_0^{kR_V}\xi\D \xi=\frac{\sqrt{2}}{3}( kR_V)^2$ for all $ \ell\in\mathbb Z_{\geq 0}$, and consequently  $ \Vert\hat\gamma_S\Vert_{L^2(O(\mathbf0,R_V);\mathbb C^3)}\leq\frac{\sqrt{2}}{3}( kR_V)^2$.

L.\ J.\ Landau has shown that \cite[\S6]{Landau2000Bessel}\begin{align}
\sup_{x\in\mathbb R}\sqrt[3]{|x|}|J_\nu(x)|\leq\sup_{x\in\mathbb R}\sqrt[3]{|x|}|J_0(x)|=0.78574+
\label{eq:Landau}\end{align}for all $ \nu\geq0$. Here, the cylindrical Bessel function $J_\nu(x) $ is related to the spherical Bessel function through the relation $ j_\ell(x)=\sqrt{\frac{\pi}{2x}}J_{\ell+\frac{1}{2}}(x)$ for $ x>0$. Thus, we have $
\int_0^{kR_V}j_\ell^2(\xi)\xi^2\D \xi\leq\frac{3}{4}(kR_V)^{4/3}
$ for all  $ \ell\in\mathbb Z_{\geq 0}$, and the same upper bound applies to $ \Vert\hat\gamma_S\Vert_{L^2(O(\mathbf0,R_V);\mathbb C^3)}$.

Now, we will set out to prove
$\int_0^{kR_V}j_\ell^2(\xi)\xi^2\D \xi\leq\frac{9}{16} kR_{V}$
 and $ \frac{1}{2\ell+1}\int_0^{kR_{V}}[(\ell+1)j_{\ell-1}^2(\xi)+\ell j_{\ell+1}^2(\xi)]\xi^2\D \xi\leq\frac{9}{16} kR_{V}$, for $ \ell\in\mathbb Z_{>0}$, using the following Lommel series  \cite[\S5.51(1)]{Watson} for $x=kR_V>0$
\begin{align}\int_0^{x}j_\ell^2(\xi)\xi^2\D \xi={}&\frac{\pi x^2}{4}\left[ J_{\ell+\frac12}^2(x)-J_{\ell-\frac12}(x)J_{\ell+\frac32}(x) \right]=\pi\sum_{L=0}^\infty\left(\ell+\frac32+2L\right)J^2_{\ell+\frac32+2L}(x)\notag\\={}&x\sum_{L=0}^\infty[2(\ell+1+2L)+1]j^2_{\ell+1+2L}(x),\label{eq:Lommel_sum}\end{align}where the cylindrical Bessel functions $J_\nu(x)$ are related to the spherical Bessel functions by $j_\ell(x)=\sqrt{\frac{\pi}{2x}}J_{\ell+\frac12}(x) $.
In view of \eqref{eq:Lommel_sum},  for any fixed $ x>0$, the two sequences $ \frac{1}{x}\int_0^{x}j_{2n}^2(\xi)\xi^2\D \xi,n\in\mathbb Z_{\geq0}$ and $\frac{1}{x}\int_0^{x}j_{2n+1}^2(\xi)\xi^2\D \xi,n\in\mathbb Z_{\geq0}$  are both monotonically non-increasing.

 When $\ell\in2\mathbb Z_{\geq0}$, we have
$\frac{1}{x}\int_0^{x}j_\ell^2(\xi)\xi^2\D \xi\leq\frac{1}{x}\int_0^{x}j_0^2(\xi)\xi^2\D \xi=\frac12-\frac{\sin 2 x}{4x}$; when  $\ell\in2\mathbb Z_{>0}$, we have
$
\frac{1}{x}\int_0^{x}j_\ell^2(\xi)\xi^2\D \xi\leq\frac{1}{x}\int_0^{x}j_2^2(\xi)\xi^2\D \xi=\frac12-\frac{\sin 2 x}{4x}-3[j_{1}(x)]^{2}$;
 when  $\ell-1\in2\mathbb Z_{\geq0}$,
we have
$\frac{1}{x}\int_0^{x}j_\ell^2(\xi)\xi^2\D \xi\leq\frac{1}{x}\int_0^{x}j_1^2(\xi)\xi^2\D \xi=\frac12+\frac{\sin 2 x}{4x} -[j_0(x)]^2$. Therefore, \begin{align}&
\sup_{\ell\in\mathbb Z_{>0}}\frac{1}{kR_V}\max\left\{\int_0^{kR_V}j_\ell^2(\xi)\xi^2\D \xi,\frac{}{}\int_0^{kR_{V}}\frac{(\ell+1)j_{\ell-1}^2(\xi)+\ell j_{\ell+1}^2(\xi)}{2\ell+1}\xi^2\D \xi \right\}\notag\\\leq{}&\sup_{\ell\in\mathbb Z_{>0}}\max\left\{\frac12+\frac{\sin (2kR_V)}{4kR_V} -[j_0(kR_V)]^2,\frac12-\frac{\sin (2kR_V)}{4kR_V}-\frac{3\ell}{2\ell+1}[j_{1}(kR_V)]^{2}\right\}\notag\\={}&\max\left\{\frac12+\frac{\sin (2kR_V)}{4kR_V} -[j_0(kR_V)]^2,\frac12-\frac{\sin (2kR_V)}{4kR_V}-[j_{1}(kR_V)]^{2}\right\}.
\end{align}

Note that\begin{align}
\frac12+\frac{\sin (2x)}{4x} -[j_0(x)]^2=\frac{1}{2}+\frac{j_0(x)\cos x}{2}-[j_0(x)]^2\leq\frac{1}{2}+\frac{|j_0(x)|}{2}-|j_0(x)|^2 \leq\frac9{16},
\end{align} where we have maximized a quadratic form with respect to $ |j_0(x)|$ in the last step.
For $ x\geq\frac{3\pi}{2}$, we have\begin{align}
\frac12-\frac{\sin (2x)}{4x}-[j_{1}(x)]^{2}\leq\frac{1}{2}+\frac1{4x}\leq\frac12+\frac{1}{6\pi}<\frac9{16};
\end{align} for $ x\in[0,\frac{\pi}{2}]\cup[\pi,\frac{3\pi}{2}]$, we have \begin{align}
\frac12-\frac{\sin (2x)}{4x}-[j_{1}(x)]^{2}\leq\frac12-\frac{\sin (2x)}{4x}\leq\frac{1}{2}<\frac{9}{16};
\end{align}for $ x\in[\frac{\pi}{2},\pi]$, we have \begin{align} {}&\frac12-\frac{\sin (2x)}{4x}-[j_{1}(x)]^{2}=\frac12-\frac{\sin x\cos x}{2x}-\left( -\frac{\cos x}{x}+\frac{\sin x}{x^2} \right)^2\notag\\\leq{}&\frac{1}{2}+\left( \frac{4}{x^{2}} -\frac{1}{2}\right)\frac{\sin x\cos x}{x}\leq\frac{1}{2}+\left( \frac{1}{2} -\frac{4}{\pi^{2}}\right)\frac{2}{\pi}<\frac{9}{16},\end{align}thereby arriving at our goal.
 \end{proof}
\subsection{$ L^2$-norm bounds of $  \Vert\R\hat{\mathscr G}\Vert$ for large spheres\label{subsec:kRgammaC}}
Arguing as in Remark \ref{rmk:V_mono}, we can verify that $ \Vert\R\hat{\mathscr G}\Vert_{L^2(V;\mathbb C^3)}\leq \Vert\R\hat{\mathscr G}\Vert_{L^2(O(\mathbf 0,R_V);\mathbb C^3)}$. Therefore, to improve the bound estimates for $ \Vert\R\hat{\mathscr G}\Vert_{L^2(V;\mathbb C^3)}$, we will focus on spherical scatterers $ V=O(\mathbf0,R)$.

Without compounding the validity of our quantitative analysis, we may momentarily regard $\bm E\in C^\infty_0(V;\mathbb C^3)$ as a vector field with compact support,
thanks to   $\Vert\R\hat{\mathscr G}\Vert_{L^2(V;\mathbb C^3)}=\sup_{\bm E\in C^\infty_0(V;\mathbb C^3)\smallsetminus\{\mathbf0\}}\left|\frac{\langle\bm E,(\R\hat{\mathscr G})\bm E\rangle_{V}}{\langle\bm E,\bm E\rangle_{V}} \right| $. Accordingly,  $ (\R\hat{\mathscr G})\bm E\in C^{\infty}(\mathbb R^3;\mathbb C^3)$ is a smooth vector field defined on the entire space.  Now that   $\bm E\in C^\infty_0(V;\mathbb C^3)\subset\mathscr S(\mathbb R^3;\mathbb C^3)$ is a member of the Schwartz space,  the Fourier transform $\widetilde{\bm E}(\bm q) $  will decay sufficiently fast as $ |\bm q|\to+\infty$. This allows us to deduce\footnote{This can also be interpreted as a formal implementation of the Hilbert transform  $ \frac{2}{\pi}\mathscr P\int_0^{+\infty} \frac{q\sin (q\xi)}{q^{2}-k^2} \D q=\cos (k\xi)$ for $ \xi=|\bm r-\bm r'|$ \cite[p.~533, (13D.7)]{KingVol2}.}
\begin{align}&\langle\bm E,(\R\hat{\mathscr G})\bm E\rangle_{V}+\langle\bm E,\bm E\rangle_{V}\notag\\={}&\lim_{\varepsilon\to0^+}\R\iiint_{\mathbb R^3}\left[\frac{1}{|\bm q|^2-(k-i\varepsilon)^2}+\frac{1}{|\bm q|^2-(k+i\varepsilon)^2}\right]\frac{|\bm q\times\widetilde{\bm E}(\bm q)|^{2}}{2}\frac{\D^3\bm q}{(2\pi)^3}\notag\\={}&\frac{2}{\pi}\mathscr P\int_0^{+\infty} \frac{q\sigma_{\bm E}(q)}{q^{2}-k^2} \D q,\end{align}where $ \mathscr P$ stands for Cauchy principal value and $ \sigma_{\bm E}(q)\colonequals\frac{q^{3}}{16\pi^{2}}\varoiint_{|\bm
n|=1}|\bm n\times\widetilde{\bm E}(q\bm n)|^2\D\Omega$ [cf.~\eqref{eq:gammaS}]. We further split this integral into two parts:\begin{align}
\frac{1}{\pi}\mathscr P\int_{k-\varepsilon}^{k+\varepsilon}\frac{\sigma_{\bm E}(q)}{q-k}\D q=\frac{1}{\pi}\int_{k-\varepsilon}^{k+\varepsilon}\frac{\sigma_{\bm E}(q)-\sigma_{\bm E}(k)}{q-k}\D q\label{eq:CPV_sigmaE}
\end{align}and \begin{align}
\frac{1}{\pi}\left(\int_0^{k-\varepsilon}+\int_{k+\varepsilon}^{+\infty}\right)\frac{\sigma_{\bm E}(q)}{q-k}\D q+\frac{1}{\pi}\int_{0}^{+\infty}\frac{\sigma_{\bm E}(q)}{q+k}\D q.
\end{align}Here in \eqref{eq:CPV_sigmaE}, we have made use of the fact that $ \mathscr P\int_{k-\varepsilon}^{k+\varepsilon}\frac{\D q}{q-k}=0$. By the first mean value theorem for integration and Lagrange's mean value theorem for differentiation, we can identify
 \eqref{eq:CPV_sigmaE} with\begin{align}
\frac{2\varepsilon}{\pi}\left.\frac{\D \sigma_{\bm E}(q)}{\D q}\right|_{q=\kappa}=\frac{2\varepsilon}{\pi}\iiint_V\left[\nabla\times \nabla\times \iiint_V\frac{\bm E^*(\bm r' )\cos(\kappa|\bm r-\bm r'|)}{4\pi}  \D^3 \bm r'\right]\cdot\bm E(\bm r)\D^3\bm r\label{eq:sigma'_kappa}
\end{align}for a certain $ \kappa\in(k-\varepsilon,k+\varepsilon)$. This naturally motivates a quantitative study of the following operator defined for any  $k>0$:\begin{align}
(\hat g \bm E)(\bm r)\colonequals\nabla\times \nabla\times \iiint_V\frac{\bm E(\bm r' )\cos( k|\bm r-\bm r'|)}{4\pi}  \D^3 \bm r'=\frac{\partial}{\partial k}(\hat\gamma_S\bm E)(\bm r).
\end{align}\begin{lemma}    Define   $ q_n^{\mathrm{TE}}(x)\colonequals x^{3}[j_n(x)]^2,n\in\mathbb Z_{\geq0}$ and $ q_n^{\mathrm{TM}}(x)\colonequals\frac{n+1}{2n+1}q_{n-1}^{\mathrm{TE}}(x)+\frac{n}{2n+1}q_{n+1}^{\mathrm{TE}}(x),n\in\mathbb Z_{>0}$. Take $ \bm f^{\mathrm{TE}}_{\ell m}$, $ \bm f^{\mathrm{TM}}_{\ell m}$, $ \bm g^{\mathrm{TE}}_{\ell m}$ and $ \bm g^{\mathrm{TM}}_{\ell m}$  as defined in Proposition \ref{prop:Mie_proj}.  An orthogonal basis set for the range of $\hat g\colon L^2(O(\mathbf 0,R);\mathbb C^3)\longrightarrow L^2(O(\mathbf 0,R);\mathbb C^3)$ is  formed by all the eigenmodes in the following eigenequations:
\begin{align}\hat g\bm h_{\ell m,+}^{\mathrm{TE}}=\lambda_{\ell m,+}^{\mathrm{TE}}\bm h_{\ell m,+}^{\mathrm{TE}},\quad\hat g\bm h_{\ell m,-}^{\mathrm{TE}}=\lambda_{\ell m,-}^{\mathrm{TE}}\bm h_{\ell m,-}^{\mathrm{TE}},\label{eq:g_TE}\\\hat g\bm h_{\ell m,+}^{\mathrm{TM}}=\lambda_{\ell m,+}^{\mathrm{TM}}\bm h_{\ell m,+}^{\mathrm{TM}},\quad \hat g\bm h_{\ell m,-}^{\mathrm{TM}}=\lambda_{\ell m,-}^{\mathrm{TM}}\bm h_{\ell m,-}^{\mathrm{TM}},\label{eq:g_TM}\end{align}
where \begin{align}
\lambda_{\ell m,\pm}^{\mathrm{TE}}={}&\frac{q_\ell^{\mathrm{TE}}(kR)\pm
\sqrt{\left[q_\ell^{\mathrm{TE}}(kR)\right]^2+\frac{4k^{8}\Vert\bm f^{\mathrm{TE}}_{\ell m}\Vert^2_{L^2(O(\mathbf 0,R);\mathbb C^3)}\Vert(\hat I-\hat P)\bm g^{\mathrm{TE}}_{\ell m}\Vert^2_{L^2(O(\mathbf 0,R);\mathbb C^3)}}{\ell^{2}(\ell+1)^{2}}}}{2k},
\\
\lambda_{\ell m,\pm}^{\mathrm{TM}}={}&\frac{q_\ell^{\mathrm{TM}}(kR)\pm
\sqrt{\left[q_\ell^{\mathrm{TM}}(kR)\right]^2+\frac{4k^{4}\Vert\bm f^{\mathrm{TM}}_{\ell m}\Vert^2_{L^2(O(\mathbf 0,R);\mathbb C^3)}\Vert(\hat I-\hat P)\bm g^{\mathrm{TM}}_{\ell m}\Vert^2_{L^2(O(\mathbf 0,R);\mathbb C^3)}}{\ell^{2}(\ell+1)^{2}}}}{2k},\\
\bm h_{\ell m,\pm}^{\mathrm{TE}}={}&\frac{k^{3}}{\ell(\ell+1)}\Vert(\hat I-\hat P)\bm g^{\mathrm{TE}}_{\ell m}\Vert^2_{L^2(O(\mathbf 0,R);\mathbb C^3)}\bm f^{\mathrm{TE}}_{\ell m}+\lambda_{\ell m,\mp}^{\mathrm{TE}}(\hat I-\hat P)\bm g^{\mathrm{TE}}_{\ell m},\\
\bm h_{\ell m,\pm}^{\mathrm{TM}}={}&\frac{k}{\ell(\ell+1)}\Vert(\hat I-\hat P)\bm g^{\mathrm{TM}}_{\ell m}\Vert^2_{L^2(O(\mathbf 0,R);\mathbb C^3)}\bm f^{\mathrm{TM}}_{\ell m}+\lambda_{\ell m,\mp}^{\mathrm{TM}}(\hat I-\hat P)\bm g^{\mathrm{TM}}_{\ell m},\label{eq:h_TM}
\end{align}for $ \ell\in\mathbb Z_{>0},m\in\mathbb Z\cap[-\ell,\ell]$. Consequently, we have \begin{align}
\Vert \hat g\Vert_{L^2(O(\mathbf 0,R);\mathbb C^3)}={}&\sup_{\ell\in\mathbb Z_{>0}}\max\{\lambda_{\ell 0,+}^{\mathrm{TE}},\lambda_{\ell 0,+}^{\mathrm{TM}}\}\leq\frac{\frac{9(kR)^2}{16}+\frac{(kR)^{4/3}}{2}+\frac{2(kR)^{2/3}}{9}}{k}.
\label{eq:g_norm_bd}\end{align}\end{lemma}

\begin{proof}
   Differentiating both sides of \eqref{eq:gamma_S_flm}  with respect to $k$,\footnote{The interchange of differentiation and summation is justified by rapid decays of the summands for large $ \ell$, thanks to $ j_\ell(kr)\sim\frac{(kr)^{\ell}}{(2\ell+1)!!},\ell\to+\infty$ and $ \frac{\partial}{\partial k}j_\ell(kr)=\frac{\ell}{k}j_\ell(kr)-rj_{\ell+1}(kr)$.} before specializing to $ \bm F=\bm E-\hat P \bm E$, we arrive at  \begin{align}&
\hat g(\bm E-\hat P\bm E)\notag\\={}&\sum_{\ell=1}^\infty\sum_{m=-\ell}^\ell\frac{k}{\ell(\ell+1)}\left(k^{2}\left\langle\frac{\partial\bm f^{\mathrm{TE}}_{\ell m}}{\partial k},\bm E-\hat P\bm E\right\rangle _V\bm f^{\mathrm{TE}}_{\ell m}+\left\langle\frac{\partial\bm f^{\mathrm{TM}}_{\ell m}}{\partial k},\bm E-\hat P\bm E\right\rangle _V \bm f^{\mathrm{TM}}_{\ell m}\right)\notag\\={}&-\sum_{\ell=1}^\infty\sum_{m=-\ell}^\ell\frac{k}{\ell(\ell+1)}\left(k^{2}\left\langle(\hat I-\hat P)\bm g^{\mathrm{TE}}_{\ell m},\bm E\right\rangle _V\bm f^{\mathrm{TE}}_{\ell m}+\left\langle(\hat I-\hat P)\bm g^{\mathrm{TM}}_{\ell m},\bm E\right\rangle _V \bm f^{\mathrm{TM}}_{\ell m}\right),\label{eq:g_gP}
\end{align}where $V= O(\mathbf 0,R)$. Owing to the orthogonality relations \eqref{eq:glm_TE_gl'm'TE}--\eqref{eq:glm_TE_gl'm'TM}, we have \begin{align}
\hat g(\hat I-\hat P)\bm g^{\mathrm{TE}}_{\ell m}={}&-\frac{k^{3}\Vert (\hat I-\hat P)\bm g^{\mathrm{TE}}_{\ell m}\Vert^2_{L^2(V;\mathbb C^3)}}{\ell(\ell+1)}\bm f^{\mathrm{TE}}_{\ell m},\label{eq:g-gP_TE}\\\hat g(\hat I-\hat P)\bm g^{\mathrm{TM}}_{\ell m}={}&-\frac{k\Vert (\hat I-\hat P)\bm g^{\mathrm{TM}}_{\ell m}\Vert^2_{L^2(V;\mathbb C^3)}}{\ell(\ell+1)}\bm f^{\mathrm{TM}}_{\ell m}.\label{eq:g-gP_TM}
\end{align}

We may reexamine the first equality in \eqref{eq:g_gP} from a different perspective. Noting that  the real part of \eqref{eq:BornGexpn} leaves us $(\R\hat{\mathscr G})\bm f^{\mathrm{TE}}_{\ell m}=\frac{k}{2} \frac{\partial\bm f^{\mathrm{TE}}_{\ell m}}{\partial k}+\mu_{\ell }^{\mathrm{TE}}\bm f^{\mathrm{TE}}_{\ell m}$ for a certain $ \mu_{\ell}^{\mathrm{TE}}\in\mathbb R$,  we may argue that \begin{align}
\frac{k^3}{\ell(\ell+1)}\left\langle\frac{\partial\bm f^{\mathrm{TE}}_{\ell m}}{\partial k},\bm E-\hat P\bm E\right\rangle _V\bm f^{\mathrm{TE}}_{\ell m,0}={}&\frac{2}{k}\frac{\langle (\R\hat{\mathscr G})\bm f^{\mathrm{TE}}_{\ell m},\bm E-\hat P\bm E\rangle _V\hat\gamma_{S}\bm f^{\mathrm{TE}}_{\ell m}}{\langle\bm f^{\mathrm{TE}}_{\ell m},\bm f^{\mathrm{TE}}_{\ell m}\rangle _V}\notag\\={}&\frac{2}{k}\frac{\langle\bm f^{\mathrm{TE}}_{\ell m}, \hat\gamma_{S}(\R\hat{\mathscr G})(\hat I-\hat P)\bm E\rangle _V\bm f^{\mathrm{TE}}_{\ell m}}{\langle\bm f^{\mathrm{TE}}_{\ell m},\bm f^{\mathrm{TE}}_{\ell m}\rangle _V}
\end{align}and similarly for the TM wave modes. Thus, we have $ \hat g(\hat I-\hat P)=\frac{2}{k}\hat P \hat\gamma_{S}(\R\hat{\mathscr G})(\hat I-\hat P)$.
Furthermore, as we have $ (\R\hat{\mathscr G})\hat\gamma_{S}(\hat I-\hat P)\bm E=\mathbf 0$, the aforementioned relation  can also be written as $ \hat g(\hat I-\hat P)=\frac{2}{k}\hat P[\hat\gamma_S,\R\hat{\mathscr G}](\hat I-\hat P)$, with a commutator $ [\hat A,\hat B]\colonequals\hat A\hat B-\hat B\hat A$.

Meanwhile, differentiating both sides of \eqref{eq:gammaS_TE} with respect to $k$, we obtain\begin{align}&
\hat g \bm f^{\mathrm{TE}}_{\ell m}+\frac{2}{k} \hat\gamma_S\left[(\R\hat{\mathscr G})\bm f^{\mathrm{TE}}_{\ell m} -\mu_{\ell m}^{\mathrm{TE}}\bm f^{\mathrm{TE}}_{\ell m}\right]\notag\\={}& R[j_\ell(kR)kR]^2\bm f^{\mathrm{TE}}_{\ell m}+\frac{2}{k} \left[(\R\hat{\mathscr G})\hat\gamma_S\bm f^{\mathrm{TE}}_{\ell m} - \mu_{\ell m}^{\mathrm{TE}}\hat\gamma_S\bm f^{\mathrm{TE}}_{\ell m}\right],
\end{align} which is equivalent to $ (\hat g+\frac{2}{k}[\hat\gamma_{S},\R\hat{\mathscr G}])\bm f^{\mathrm{TE}}_{\ell m} =R[j_\ell(kR)kR]^2\bm f^{\mathrm{TE}}_{\ell m}$. Reasoning similarly for TM wave modes, we are led to $ \hat g\hat P=-\frac{2}{k}[\hat\gamma_{S},\R\hat{\mathscr G}]\hat P+\hat Q$, where\begin{align}\hat Q\bm E\colonequals{}&\frac1k\sum_{\ell=1}^\infty\sum_{m=-\ell}^\ell\left[q_\ell^{\mathrm{TE}}(kR)\frac{\langle\bm f^{\mathrm{TE}}_{\ell m},\bm E\rangle _V\bm f^{\mathrm{TE}}_{\ell m}}{\langle\bm f^{\mathrm{TE}}_{\ell m},\bm f^{\mathrm{TE}}_{\ell m}\rangle _V}+q_\ell^{\mathrm{TM}}(kR)\frac{\langle\bm f^{\mathrm{TM}}_{\ell m},\bm E\rangle _V\bm f^{\mathrm{TM}}_{\ell m}}{\langle\bm f^{\mathrm{TM}}_{\ell m},\bm f^{\mathrm{TM}}_{\ell m}\rangle _V}\right].\label{eq:g_Q}
\end{align} Since the left-hand side of  $ \hat P\hat g\hat P-\hat Q=-\frac{2}{k}\hat P[\hat\gamma_{S},\R\hat{\mathscr G}]\hat P$ defines a Hermitian operator, so must  its right-hand side. However,  from the relation   $ ( \hat P[\hat\gamma_S,\R\hat{\mathscr G}]\hat P)^*= \hat P([\hat\gamma_S,\R\hat{\mathscr G}])^*\hat P=- \hat P[\hat\gamma_S,\R\hat{\mathscr G}]\hat P$, we see that the Hermitian operator $ \hat P[\hat\gamma_S,\R\hat{\mathscr G}]\hat P$ is identically vanishing. Therefore, we have $g\hat P=(\hat I-\hat P)\hat g+\hat Q$. In particular, the transposes of \eqref{eq:g-gP_TE} and \eqref{eq:g-gP_TM} bring us \begin{align}\hat
g\bm f^{\mathrm{TE}}_{\ell m}={}&\frac{q_\ell^{\mathrm{TE}}(kR)}{k}\bm f^{\mathrm{TE}}_{\ell m}-\frac{k^{3}\Vert\bm f^{\mathrm{TE}}_{\ell m}\Vert^2_{L^2(V;\mathbb C^3)}}{\ell(\ell+1)}(\hat I-\hat P)\bm g^{\mathrm{TE}}_{\ell m},\label{eq:gP_TE}\\\hat g\bm f^{\mathrm{TM}}_{\ell m}={}&\frac{q_\ell^{\mathrm{TM}}(kR)}{k}\bm f^{\mathrm{TM}}_{\ell m}-\frac{k\Vert\bm f^{\mathrm{TM}}_{\ell m}\Vert^2_{L^2(V;\mathbb C^3)}}{\ell(\ell+1)}(\hat I-\hat P)\bm g^{\mathrm{TM}}_{\ell m}.\label{eq:gP_TM}
\end{align} Diagonalization of \eqref{eq:g-gP_TE}, \eqref{eq:g-gP_TM}, \eqref{eq:gP_TE} and \eqref{eq:gP_TM} then leads us to \eqref{eq:g_TE}--\eqref{eq:h_TM}.

It is  clear from Proposition \ref{prop:Mie_proj} that $ \lambda_{\ell m,+}^{\mathrm {TE}}=\lambda_{\ell0,+}^{\mathrm {TE}}\geq|\lambda_{\ell m,-}^{\mathrm {TE}}|>0$ and  $ \lambda_{\ell m,+}^{\mathrm {TM}}=\lambda_{\ell0,+}^{\mathrm {TM}}\geq|\lambda_{\ell m,-}^{\mathrm {TM}}|>0$. Moreover, the bound estimates in \eqref{eq:gTE_bd}, \eqref{eq:gTM_bd}, \eqref{eq:gammaS_bd} and \eqref{eq:Landau} bring us  \begin{align}
\max\{\lambda_{\ell 0,+}^{\mathrm{TE}},\lambda_{\ell 0,+}^{\mathrm{TM}}\}\leq\frac{(kR)^{4/3}+\sqrt{(kR)^{8/3}+4\left(\frac{3kR}{4}\right)^{4}}}{2k}\leq\frac{\frac{9(kR)^2}{16}+\frac{(kR)^{4/3}}{2}+\frac{2(kR)^{2/3}}{9}}{k},
\end{align}
as claimed in \eqref{eq:g_norm_bd}.
   \end{proof}
\begin{proposition}\label{prop:kRVgammaC}For $V=O(\mathbf0,R)$,  we have \begin{align}
\Vert\R\hat{\mathscr G}\Vert_{L^2(V;\mathbb C^3)}\leq{}&2+\frac{4 k R}{5}+\frac{1}{2\pi} \left(  \sqrt[3]{k R+\frac{1}{2}}+\frac{4}{9 \sqrt[3]{k R+\frac{1}{2}}} \right),\label{eq:ReG_bd}
\end{align}and the same bound applies to $ \Vert\R\hat{\gamma}\Vert_{L^2(V;\mathbb C^3)}=\Vert\hat{\gamma}_{C}\Vert_{L^2(V;\mathbb C^3)}$.\end{proposition}
\begin{proof}By \eqref{eq:sigma'_kappa} and  \eqref{eq:g_norm_bd}, we may put down\begin{align}
\left\vert\frac{1}{\pi}\mathscr P\int_{k-\varepsilon}^{k+\varepsilon}\frac{\sigma_{\bm E}(q)}{q-k}\D q\right\vert\leq\frac{2\varepsilon\langle\bm E,\bm E\rangle_{V}}{\pi}\left[ \frac{9(k+\varepsilon )R^2}{16}+\frac{(k+\varepsilon )^{1/3}R^{4/3}}{2}+\frac{2R^{2/3}}{9(k+\varepsilon )^{1/3}} \right]
\end{align}for $ \varepsilon\in(0,k)$, where we have additionally exploited the fact that $ \frac{(kR)^{4/3}+\sqrt{(kR)^{8/3}+4\left(\frac{3kR}{4}\right)^{4}}}{2k}$ is monotonically increasing in $k$. For $ \delta\in(\varepsilon,k)$, the inequality $ 0\leq\sigma_{\bm E}(q)\leq\frac{9qR}{16}\langle\bm E,\bm E\rangle_{V}$ [cf.\ \eqref{eq:gammaS_bd}] allows us to estimate\begin{align}
\left\vert\frac{1}{\pi}\left(\int_{k-\delta}^{k-\varepsilon}+\int_{k+\varepsilon}^{k+ \delta}\right)\frac{\sigma_{\bm E}(q)}{q-k}\D q\right\vert
\leq\frac{9R\langle\bm E,\bm E\rangle_{V}}{16\pi}\int_{k+\varepsilon}^{k+ \delta}\frac{q\D q}{q-k}= \frac{9R\langle\bm E,\bm E\rangle_{V}}{16\pi}\left( \delta-\varepsilon+k\log\frac{\delta}{\varepsilon} \right).\end{align}  In the meantime, the inequality $ 0\leq\sigma_{\bm E}(q)\leq\frac{q^{3}}{16\pi^{2}}\varoiint_{|\bm
n|=1}|\widetilde{\bm E}(q\bm n)|^2\D\Omega$ and the Parseval--Plancherel identity $\langle\bm E,\bm E\rangle_{V}=\frac{1}{(2\pi)^3}\int_0^{+\infty}[\varoiint_{|\bm
n|=1}|\widetilde{\bm E}(q\bm n)|^2\D\Omega]q^2\D q$  leave us the bounds \begin{align}
\left\vert\frac{1}{\pi}\left(\int_0^{k-\delta}+\int_{k+\delta}^{2k}\right)\frac{\sigma_{\bm E}(q)}{q-k}\D q\right\vert\leq\frac{2k}{\delta\pi}\int_{0}^{+\infty}\frac{\sigma_{\bm E}(q)}{q}\D q\leq\frac{k}{\delta}\langle\bm E,\bm E\rangle_{V},\label{eq:kdelta}
\end{align}\begin{align}0\leq
\frac{1}{\pi}\int_{2k}^{+\infty}\frac{\sigma_{\bm E}(q)}{q-k}\D q\leq\frac{2}{\pi}\int_{2k}^{+\infty}\frac{\sigma_{\bm E}(q)}{q}\D q\leq\frac{2}{\pi}\int_{0}^{+\infty}\frac{\sigma_{\bm E}(q)}{q}\D q=\langle\bm E,\bm E\rangle_{V},
\end{align}and\begin{align}
0\leq \frac{1}{\pi}\int_{0}^{+\infty}\frac{\sigma_{\bm E}(q)}{q+k}\D q\leq \frac{1}{\pi}\int_{0}^{+\infty}\frac{\sigma_{\bm E}(q)}{q}\D q=\frac{\langle\bm E,\bm E\rangle_{V}}{2},
\end{align}so the following bound estimate holds  for any $ \varepsilon\in(0,k),\delta\in(\varepsilon,k)$: \begin{align}
|\langle\bm E,(\R\hat{\mathscr G})\bm E\rangle_{V}|\leq{}&\frac{2\varepsilon\langle\bm E,\bm E\rangle_{V}}{\pi}\left[ \frac{9(k+\varepsilon )R^2}{16}+\frac{(k+\varepsilon )^{1/3}R^{4/3}}{2}+\frac{2R^{2/3}}{9(k+\varepsilon )^{1/3}} \right]\notag\\{}&+\frac{9R\langle\bm E,\bm E\rangle_{V}}{16\pi}\left( \delta-\varepsilon+k\log\frac{\delta}{\varepsilon} \right)+\frac{k}{\delta}\langle\bm E,\bm E\rangle_{V}+\langle\bm E,\bm E\rangle_{V}.
\end{align}In particular, setting  $ \delta=\frac{16 \pi }{9 R}$ and $\varepsilon=\frac{1}{2 R}$ to minimize $\frac{9kR}{16\pi}\log\delta+\frac{k}{\delta}$ and $ \frac{2\varepsilon}{\pi}\frac{9kR^2}{16}+\frac{9kR}{16\pi}\log\frac{1}{\varepsilon}$, we obtain \begin{align}
|\langle\bm E,(\R\hat{\mathscr G})\bm E\rangle_{V}|\leq \left[2+\frac{9 k R}{16\pi}\left( 2+\log\frac{32\pi}{9} \right)+\frac{1}{2\pi} \left(  \sqrt[3]{k R+\frac{1}{2}}+\frac{4}{9 \sqrt[3]{k R+\frac{1}{2}}} \right)\right]\langle\bm E,\bm E\rangle_{V}.
\end{align}Since $ \frac{9}{16\pi}\left( 2+\log\frac{32\pi}{9} \right)=0.79018+$, we arrive at \eqref{eq:ReG_bd} under the condition that  $ kR>\frac{16\pi}{9} $. For $ \frac{1}{2}<kR\leq\frac{16\pi}{9}$, one needs to  set   $ \delta=k$ and $\varepsilon=\frac{1}{2 R}$ while ignoring the $ \frac{k}{\delta}\langle\bm E,\bm E\rangle_{V}$ term  due to \eqref{eq:kdelta}, to deduce\begin{align}
|\langle\bm E,(\R\hat{\mathscr G})\bm E\rangle_{V}|\leq \left[1+\frac{9 k R}{8\pi}+\frac{9 k R}{16\pi}\log(2kR)+\frac{1}{2\pi} \left(  \sqrt[3]{k R+\frac{1}{2}}+\frac{4}{9 \sqrt[3]{k R+\frac{1}{2}}} \right)\right]\langle\bm E,\bm E\rangle_{V},
\end{align}where the bracketed expression does not exceed the right-hand side of  \eqref{eq:ReG_bd}.  For $ 0<kR\leq\frac{1}{2}$, Proposition \ref{prop:AreaBound} guarantees that $ \Vert\R\hat{\mathscr G}\Vert_{L^2(V;\mathbb C^3)}\leq1+\frac{3(kR)^{2} }{5\pi}\left(\frac{4\pi}{3}\right)^{2/3}$, which is  more stringent than the right-hand side of  \eqref{eq:ReG_bd}  in this regime.

One may adapt the analysis above to $ \langle\bm E,\hat \gamma_C\bm E\rangle_{V}=\frac{2}{\pi}\mathscr P\int_0^{+\infty} \frac{q\sigma_{\bm E}(q)}{q^{2}-k^2} \D q-\frac{2}{\pi}\int_{0}^{+\infty}\frac{\sigma_{\bm E}(q)}{q}\D q$.  \end{proof}

Combining the  results above with the conclusions from \S\ref{subsec:gamma_L2_bd}, we immediately arrive at  a  bound estimate for   $\Vert\hat{ \mathscr G}\Vert_{L^2(V;\mathbb C^3)}\leq\Vert\R\hat{ \mathscr G}\Vert_{L^2(V;\mathbb C^3)}+\Vert\I\hat{ \mathscr G}\Vert_{L^2(V;\mathbb C^3)}$ that combines \eqref{eq:Re_gamma_norm_bd} and \eqref{eq:Im_gamma_norm_bd}, as well as a  bound  estimate for  $\Vert\hat{ \gamma}\Vert_{L^2(V;\mathbb C^3)}$ stated in \eqref{eq:gamma_norm_bd_statement}.

\subsection{Hilbert--Schmidt bounds for $   \Vert\hat\gamma\Vert_2$\label{subsec:HS_bd}}The right-hand side of \eqref{eq:quad_surf_int} is manifestly invariant under a similarity transformation of the boundary surface $ \partial V$.
Thus, if we fix the wavelength $ 2\pi/k$ and consider a family of similar shapes of dielectrics, then the dominant contribution to the Hilbert--Schmidt bound of $\Vert (\hat I+2\hat{\mathscr G}-2\hat\gamma)^2(\hat{\mathscr G}-\hat\gamma)\Vert_2$ for large volumes will come from $ \Vert\hat \gamma\Vert_2$.

To prove the upper bound estimate of  $ \Vert\hat \gamma\Vert_2$ stated in Theorem~\ref{thm:HS_norm}, we need some preparations.

\begin{lemma}
For
the  $3\times3$  matrix $\hat{\Gamma}(\bm r,\bm r')$ defined in \eqref{eq:Gamma_mat}, we have\begin{align}\Tr[\hat {\Gamma}^*(\bm r,\bm r')\hat {\Gamma}(\bm r,\bm r')]\leq{}&\frac{9k^4}{64\pi^2|\bm r-\bm r' |^2}.\end{align} \end{lemma}\begin{proof}
We note that for $a,b\in\mathbb C $, there is an algebraic identity $\Tr[(a\mathbf1+b\hat{\bm r}\hat{\bm r}^\mathrm{T})^*(a\mathbf1+b\hat{\bm r}\hat{\bm r}^\mathrm{T})]=\Tr[|a|^{2}\mathbf1+2\R(a^{*}b)\hat{\bm r}\hat{\bm r}^\mathrm{T}+|b|^{2}\hat{\bm r}\hat{\bm r}^\mathrm{T}]=3|a|^{2}+2\R(a^{*}b)+|b|^{2}$. In particular, this brings us $ \Tr[\hat {\Gamma}^*(\bm r,\bm r')\hat {\Gamma}(\bm r,\bm r')]=\left(\frac{k^{2}}{4\pi |\bm r-\bm r'|}\right)^2f(k|\bm r-\bm r'|)$ where $ f(\xi)\colonequals\frac{2}{\xi^{4}}[\xi ^4+\xi ^2+2 (\xi ^2-3) \cos \xi +6(1- \xi  \sin \xi)]$.

For all $ \xi>0$, we can establish a  bound estimate $ f(\xi)\leq\frac94$ on Taylor expansions with Lagrange  remainders and routine maximization procedures for Laurent polynomials in the variable $\xi$. Concretely speaking, for $ \xi\in(0,\sqrt{3}]$, we have $
f(\xi)\leq{}\frac{2}{\xi^{4}}\big[\xi ^4+\xi ^2+2 (\xi ^2-3)\big( 1-\frac{\xi^{2}}{2} \big) +6- 6\xi  \big( \xi-\frac{\xi^3}{6}\big)\big]=2<\frac{9}{4}$;
for $ \xi\in[\sqrt{3},\pi]$, we have $ f(\xi)\leq f_{1}(\xi)\colonequals\frac{2}{\xi^{4}}\big[\xi ^4+\xi ^2+2 (\xi ^2-3)\frac{(\xi -\pi )^2-2}{2}+6- 6\xi(\xi-\pi)\frac{(\xi -\pi )^2-6}{6}  \big]\leq  f_{1}(\pi)=f(\pi)<\frac{9}{4}$;
for $ \xi\in[\pi,\frac{3 \pi }{2}]$, we have $ f(\xi)\leq f_2(\xi)\colonequals\frac{2}{\xi^{4}}\big\{\xi ^4+\xi ^2+2 (\xi ^2-3)\big[\left(\xi -\frac{3 \pi }{2}\right)-\frac{1}{6} \left(\xi -\frac{3 \pi }{2}\right)^3\big]+6+ 6\xi  \big[1-\frac{1}{2} \left(\xi -\frac{3 \pi }{2}\right)^2+\frac{1}{24} \left(\xi -\frac{3 \pi }{2}\right)^4\big]\big\}\leq f_{2}(\frac{3 \pi }{2})=f(\frac{3 \pi }{2})<\frac{9}{4}$; for $ \xi\in[\frac{3 \pi }{2},2(\sqrt[3]2+\sqrt[3]4)]$, we have $ f(\xi)\leq f_{3}(\xi)\colonequals\frac{2}{\xi^{4}}[\xi ^4+\xi ^2+2 (\xi ^2-3)(\xi -\frac{3 \pi }{2} )+6(1+ \xi)]\leq f_{3}(2(\sqrt[3]2+\sqrt[3]4))<\frac{9}{4}$; for $ \xi\geq2(\sqrt[3]2+\sqrt[3]4)$, we have $
f(\xi)\leq{}\frac{2}{\xi^{4}}\left[\xi ^4+\xi ^2+2 (\xi ^2-3) +6(1+\xi)  \right]=2+\frac{6}{\xi ^2}+\frac{12}{\xi ^3}\leq2+\frac{6}{[2(\sqrt[3]2+\sqrt[3]4)] ^2}+\frac{12}{[2(\sqrt[3]2+\sqrt[3]4)] ^3}=\frac{9}{4}$.    \end{proof}
\begin{remark}Unlike what we did in \S\S\ref{subsec:gamma_L2_bd}--\ref{subsec:kRgammaC}, we do not need separate treatments of $ \Tr[\hat {\Gamma}^*(\bm r,\bm r')\hat {\Gamma}(\bm r,\bm r')]$ for small and large scatterers. This is because both  $ \lim_{\xi\to0}f(\xi)=\frac32$ and $ \lim_{\xi\to+\infty}f(\xi)=2$ are close to the value of $ \max_{\xi>0}f(\xi)=2.2315+$.
\eor\end{remark}

Now we may turn to the functional
\begin{align} I[V]=\sqrt{\iiint_{ V}\left\{ \iiint_V \frac{1}{|\bm r-\bm r' |^2}\D^3\bm r'\right\}\D^3\bm r},\end{align}
so that we have  the Hilbert--Schmidt bound\begin{align}   \Vert\hat \gamma\Vert_{2}\leq\sqrt{\iiint_{ V}\left\{ \iiint_V \Tr[\hat {\Gamma}^*(\bm r,\bm r')\hat {\Gamma}(\bm r,\bm r')]\D^3\bm r'\right\}\D^3\bm r}=\frac{3k^2}{8\pi}I[V],\end{align}
\begin{lemma}
We have
\begin{align}I[V]\leq\left(\sqrt2\pi\iiint_V\D^3\bm r\right)^{2/3}\end{align}
and consequently, we have  $\Vert\hat \gamma\Vert_{2}\leq\frac{1}{3}\left(k^3\iiint_V\D^3\bm r\right)^{2/3}$.\end{lemma}\begin{proof}Appealing again to  the  Hardy--Littlewood--Sobolev inequality \cite{HL1,HL2,Sob} with Lieb's sharp constant (see \cite[Theorem 3.1]{Lieb1983}, as well as \cite[Theorem~4.3]{Lieb})
\begin{align}&\left\vert\int_{\mathbb R^d}\int_{\mathbb R^d}\frac{f(\bm r')h(\bm r)}{|\bm r-\bm r'|^m}\D^d\bm r'\D^d\bm r\right\vert\notag\\\leq {}&\pi^{m/2}\frac{\Gamma(\frac{d-m}{2})}{\Gamma(d-\frac{m}{2})}\left[ \frac{\Gamma(\frac{d}{2})}{\Gamma(d)} \right]^{-1+\frac{m}{d}}\Vert f\Vert_{L^{2d/(2d-m)}(\mathbb R^d;\mathbb C)}\Vert h\Vert_{L^{2d/(2d-m)}(\mathbb R^d;\mathbb C)},\notag\\&\text{where }d=3,m=2,\frac{2d}{2d-m}=\frac{3}{2};\quad f(\bm r)=h(\bm r)=\begin{cases}1, & \bm r\in V \\
0, & \bm r\notin V \\
\end{cases},\end{align}
we may arrive at the following estimate
\begin{align}(I[V])^2=\iiint_{ V}\left\{ \iiint_V \frac{1}{|\bm r-\bm r' |^2}\D^3\bm r'\right\}\D^3\bm r\leq\left(\sqrt2\pi\iiint_V\D^3\bm r\right)^{4/3}.\end{align}Accordingly, the estimates on  $\Vert\hat \gamma\Vert_{2} $ derives from $0.322+= \frac{3}{8\pi}(\sqrt2\pi)^{2/3}\leq\frac{1}{3}$. \end{proof}\begin{remark}From this lemma, we have $ I[O(\mathbf0,R)]/ R^2\leq7.022+$, which may be compared to the exact value  $ I[O(\mathbf0,R)]/ R^2=2\pi$ obtained in the next lemma.\eor\end{remark}

\begin{lemma}
We have
\begin{align}I[V]=\sqrt{\iiint_{\mathbb R^3}\frac{|\iiint_V e^{i\bm q\cdot\bm r}\D^3\bm r|^2}{4\pi|\bm q|}\D^3\bm q}\end{align}
and in particular, we have the identity $ I[O(\mathbf0,R)]=2\pi R^2$ for all $R>0$, and $\Vert\hat \gamma\Vert_{2}\leq\frac{3}4(kR_V)^2$.
\end{lemma}
\begin{proof}We start with the Fourier inversion formula
\begin{align}\iiint_V \frac{1}{|\bm r-\bm r' |^2}\D^3\bm r'=\frac{1}{(2\pi)^3}\Fint_{\mathbb R^3}\frac{2\pi^2\iiint_V e^{i\bm q\cdot\bm r'}\D^3\bm r'}{|\bm q|}e^{-i\bm q\cdot\bm r}\D^3\bm q\end{align}where the convolution kernel satisfies\begin{align}\lim_{\varepsilon\to0^+}\iiint_{\mathbb R^3}\frac{e^{-\varepsilon|\bm r|}}{|\bm r|^2}e^{i\bm q\cdot\bm r}\D^3\bm r=4\pi\lim_{\varepsilon\to0^+}\int_0^{+\infty}\frac{e^{-\varepsilon|\bm r|}\sin(|\bm q||\bm r|)}{|\bm q||\bm r|}\D|\bm r|=\frac{2\pi^2}{|\bm q|}.\end{align}Then, applying the Parseval--Plancherel identity, we obtain the Fourier representation of $I[V]$ as claimed.

For  $V=O(\mathbf0,R)$, we have \begin{align}\iiint_{V=O(\mathbf0,R)} e^{i\bm q\cdot\bm r'}\D^3\bm r'=4\pi\frac{\sin(|\bm q|R)-|\bm q|R\cos(|\bm q|R)}{|\bm q|^3}=4\pi R^{3}\frac{j_1(|\bm q|R)}{|\bm q|R},\end{align}where $j_1(x) $ is the spherical Bessel function of the first order. Consequently, we may derive $I[O(\mathbf0,R)]=2\pi R^2 $ from the identity $\int_0^{+\infty} j_1^2(x)/x\D x=1/4 $, which is a specific case of the Weber--Schafheitlin integral \cite[pp.~402--404]{Watson}.
We may check this against  the spherical harmonic expansion
\begin{align}\frac{1}{|\bm r-\bm r' |^2}=\left|4\pi\sum_{\ell=0}^\infty\sum_{m=-\ell}^\ell\frac{Y_{\ell m}^{*}(\theta,\phi)Y_{\ell m}^{\phantom{*}}(\theta',\phi')}{(2\ell+1)\max(|\bm r|,|\bm r'|)}\left[\frac{\min(|\bm r|,|\bm r'|)}{\max(|\bm r|,|\bm r'|)}\right]^\ell\right|^2,\label{eq:Ylm_expn}\end{align}which integrates to \begin{align}&\iiint_{ O(\mathbf0,R)}\left\{ \iiint_{O(\mathbf0,R)} \frac{1}{|\bm r-\bm r' |^2}\D^3\bm r'\right\}\D^3\bm r\notag\\={}&(4\pi)^{2}\sum_{\ell=0}^\infty\frac{1}{(2\ell+1)}\int_{0 }^R\left\{ \int_{0 }^R \frac{[\min(|\bm r|,|\bm r'|)]^{2\ell}}{[\max(|\bm r|,|\bm r'|)]^{2\ell+2}}|\bm r'|^2\D|\bm r'|\right\}|\bm r|^2\D|\bm r|\notag\\={}&(4\pi)^{2}R^4\sum_{\ell=0}^\infty\frac{1}{2(2\ell+1)(2\ell+3)}=(2\pi R^2)^2.\end{align}

Suppose that we have  chosen the origin $\bm r=\mathbf 0 $ of the coordinate system to coincide with the center of the circumscribed sphere of radius $R_V$, then we have the relations $ V\subset O(\mathbf0,R_V)$, $I_{}[V]\leq I_{}[O(\mathbf0,R_V)] $, and
 this eventually establishes $\Vert\hat \gamma\Vert_{2}\leq\frac34(kR_V)^2$.
 \end{proof}

\subsection{Applications to error analysis of Born approximation}The $ L^2$-norm bounds of $ \Vert\hat \gamma\Vert_{L^2(V;\mathbb C^3)}$ established in \S\S\ref{subsec:gamma_L2_bd}--\ref{subsec:totalsc} immediately bring us the following result concerning the solvability of the Born equation in the perturbative regime. \begin{corollary}If $\chi\in\mathbb C$ satisfies the inequality
 \begin{footnotesize}\begin{align} |\chi|^2\left( 1+\min\left\{ \frac{3 }{5\pi}\left(k^3\iiint_V\D^3\bm r\right)^{2/3},1+\frac{4 k R_{V}}{5}+\frac{1}{2\pi} \left(  \sqrt[3]{k R_{V}+\frac{1}{2}}+\frac{4}{9 \sqrt[3]{k R_{V}+\frac{1}{2}}} \right)\right\} \right)<|\R\chi|,\label{ineq:gammaC_chi}\end{align}\end{footnotesize}then the Born operator $\hat{\mathscr B}=\hat I-\chi\hat{\mathscr G}\colon L^2(V;\mathbb C^3)\longrightarrow L^2(V;\mathbb C^3)$ is invertible, with a norm estimate of the inverse operator given by
\begin{small}\begin{align}\Vert  (\hat I-\chi\hat{\mathscr G})^{-1}\Vert_{L^2(V;\mathbb C^3)}  \leq {}&\frac{|\chi|}{|\R\chi|}\left[ 1-\frac{|\chi|^2}{|\R\chi|}\left( 1+\min\left\{ \frac{3 }{5\pi}\left(k^3\iiint_V\D^3\bm r\right)^{2/3},1+\frac{4 k R_{V}}{5}\right.\right.\right.\notag\\{}&\left.\left.\left.+\frac{1}{2\pi} \left(  \sqrt[3]{k R_{V}+\frac{1}{2}}+\frac{4}{9 \sqrt[3]{k R_{V}+\frac{1}{2}}} \right)  \right\} \right) \right]^{-1},\label{eq:invBorn_norm}\end{align}\end{small}where $ R_V$ is the minimal radius of all circumscribed spheres.\end{corollary}\begin{proof}For the Hermitian operator $\hat\gamma_S\colon L^2(V;\mathbb C^3)\longrightarrow L^2(V;\mathbb C^3)$, we naturally  have a resolvent bound $\Vert(\hat I+i\chi\hat\gamma_S)^{-1}\Vert_{L^2(V;\mathbb C^3)} \leq1/|i\chi\I\frac{1}{i\chi}| =\frac{|\chi|}{|\R \chi|}$ \cite[p.~210, (2)]{Yosida} so long as $ \R\chi\neq0$. Assuming  \eqref{ineq:gammaC_chi}, we have
$ 1-|\chi|P>0$  for    $ P=\Vert(\hat I+i\chi\hat\gamma_S)^{-1}\Vert_{L^2(V;\mathbb C^3)}\Vert\hat{\mathscr G}+i\hat\gamma_{S}\Vert_{L^2(V;\mathbb C^3)}\leq\frac{|\chi|}{|\R \chi|}\Vert\R\hat{\mathscr G}\Vert_{L^2(V;\mathbb C^3)}$, so that $ [\hat I-(\hat I+i\chi\hat\gamma_S)^{-1}\chi(\hat{\mathscr G}+i\hat\gamma_S)]^{-1}$ is bounded, whose  norm does not exceed $(1-|\chi|P)^{-1}$. Consequently, \begin{align} (\hat I-\chi\hat{\mathscr G})^{-1}=[\hat I-(\hat I+i\chi\hat\gamma_S)^{-1}\chi(\hat{\mathscr G}+i\hat\gamma_S)]^{-1}(\hat I+i\chi\hat\gamma_S)^{-1}\end{align}is well-defined,
and the left-hand side of \eqref{eq:invBorn_norm} is bounded by $\Vert [\hat I-(\hat I+i\chi\hat\gamma_S)^{-1}\chi(\hat{\mathscr G}+i\hat\gamma_S)]^{-1}\Vert_{L^2(V;\mathbb C^3)}\frac{|\chi|}{|\R \chi|} \leq \frac{|\chi|}{|\R \chi|}(1-|\chi|P)^{-1}$. Thus, the right-hand side of \eqref{eq:invBorn_norm} follows from Propositions \ref{prop:AreaBound}, \ref{prop:kRVgammaS} and \ref{prop:kRVgammaC}.    \end{proof}

Recalling \eqref{eq:Born_approx_err_est} and the computations in  \S\S\ref{subsec:gamma_L2_bd}--\ref{subsec:kRgammaC}, we see that the  the truncation error of Born series can be estimated with an upper bound of $ \Vert \hat{\mathscr G}\Vert_{L^2(V;\mathbb C^3)}\leq\Vert\R\hat{ \mathscr G}\Vert_{L^2(V;\mathbb C^3)}+\Vert\I\hat{ \mathscr G}\Vert_{L^2(V;\mathbb C^3)}$ stated in  \eqref{eq:Re_gamma_norm_bd} and \eqref{eq:Im_gamma_norm_bd}.
  In parallel, the total scattering cross-section satisfies
\begin{align}0\leq{}&-\I\langle \bm E,\hat{\mathscr G}\bm E\rangle_V\colonequals \langle \bm E,\hat{\gamma}_S\bm E\rangle_V\leq\Vert\hat {\gamma}_S\Vert_{L^2(V;\mathbb C^3)} \langle \bm E,\bm E\rangle_V,\end{align}where an upper bound for  $\Vert\hat {\gamma}_S\Vert_{L^2(V;\mathbb C^3)}=\Vert\I\hat{ \mathscr G}\Vert_{L^2(V;\mathbb C^3)}$ is given in \eqref{eq:Im_gamma_norm_bd}.

In practice, one usually truncates the Born series at the $m=1$ term \cite{BornWolf} due to analytical challenges and/or numerical costs of evaluating $ \hat{\mathscr G}^m\bm E_\mathrm{inc}$ for $m\geq2$. Thus, the terminology ``Born approximation'' sometimes exclusively refers to the first two leading terms in the Born series  $\bm E\sim \bm E_\mathrm{inc}+\chi\hat{\mathscr G}\bm E_\mathrm{inc}$.

For an arbitrarily shaped dielectric with real-valued  $\chi$, we now give  a brief  algebraic derivation for the error bound of total scattering cross-section in the  Born approximation. Recalling  \eqref{eq:GOT1}, we see that the true value of the total scattering cross-section is given by \begin{align} \I \langle\bm E_\mathrm{inc},-\chi k(\hat I-\chi\hat{\mathscr G})^{-1}\bm E_\mathrm{inc} \rangle_{V}=\I \langle\bm E-\chi\hat{\mathscr G}\bm E ,-\chi k\bm E \rangle_{V}=-\chi^2k\I\langle\bm E,\hat{\mathscr G}\bm E \rangle_{V}\end{align} for real-valued $\chi$. While the Born approximation leads to
\begin{align}\I \langle\bm E_\mathrm{inc},-\chi k(\hat I-\chi\hat{\mathscr G})^{-1}\bm E_\mathrm{inc} \rangle_{V}\sim{}&\I \langle\bm E_\mathrm{inc},-\chi k(\hat I+\chi\hat{\mathscr G})\bm E_\mathrm{inc} \rangle_{V}\notag\\={}&-\chi^2k\I\langle\bm E_\mathrm{inc},\hat{\mathscr G}\bm E_\mathrm{inc} \rangle_{V} .\end{align}
  The  error bound for the Born approximation of total scattering cross-section is thus given by
{\allowdisplaybreaks\begin{align}&\chi^{2}k|\I\langle\bm E_\mathrm{inc},\hat{\mathscr G}\bm E_\mathrm{inc} \rangle_{V}-\I\langle\bm E,\hat{\mathscr G}\bm E \rangle_{V}|\notag\\={}&\chi^{2}k|\I\langle\bm E-\chi\hat{\mathscr G}\bm E,\hat{\mathscr G} (\bm E-\chi\hat{\mathscr G}\bm E) \rangle_{V}-\I\langle\bm E,\hat{\mathscr G}\bm E \rangle_{V}|\notag\\={}&\chi^2k|\I\langle\hat{\mathscr G}\bm E,\chi^2\hat{\mathscr G} (\hat{\mathscr G}\bm E) \rangle_{V}-\I\langle\bm E,\chi\hat{\mathscr G} (\hat{\mathscr G}\bm E) \rangle_{V}|\notag\\={}&\chi^2k|-\langle\hat{\mathscr G}\bm E,\chi^2\hat{\gamma}_S (\hat{\mathscr G}\bm E) \rangle_{V}-\I\langle2i\hat\gamma_S\bm E,\chi(\hat{\mathscr G}\bm E) \rangle_{V}|\notag\\\leq{}&\chi^2k\left[ \chi^2\langle\hat{\mathscr G}\bm E, \hat{\mathscr G}\bm E \rangle_{V}\Vert\hat {\gamma}_S\Vert_{L^2(V;\mathbb C^3)} +2|\chi|\Vert\hat {\gamma}_S\Vert_{L^2(V;\mathbb C^3)} \Vert\hat {\mathscr G}\Vert_{L^2(V;\mathbb C^3)}\langle\bm E, \bm E \rangle_{V}\right]\notag\\\leq{}&\frac{|\chi|^{3}k\Vert\hat {\gamma}_S\Vert_{L^2(V;\mathbb C^3)} \langle\bm E_\mathrm{inc},\bm E_\mathrm{inc} \rangle_{V}}{(1-\chi\Vert\hat {\mathscr G}\Vert_{L^2(V;\mathbb C^3)})^{2}}[|\chi|\Vert\hat {\mathscr G}\Vert_{L^2(V;\mathbb C^3)}^{2}+2\Vert\hat {\mathscr G}\Vert_{L^2(V;\mathbb C^3)}].\label{eq:SCbound}\end{align}}

\subsection{Applications to Mie scattering}One can illustrate the relative  tightness of the bounds in Theorems~\ref{thm:L2_op_norm} and \ref{thm:HS_norm}  with some numerical examples.

We recall  from Proposition \ref{prop:Mie_res} that the Mie series in \eqref{eq:MieSeries} diverges when $1/\chi=\lambda\in\sigma^\Phi(\hat{\mathscr G})\smallsetminus\{0,-1/2\}\subset\sigma(\hat{\mathscr G})$ represents a Mie resonance mode.
When we consider $ \sigma^\Phi(\hat{\mathscr G})$ for spherical scatterers, there is a refinement of \eqref{eq:Re_gamma_norm_bd}, as expounded in the next lemma. \begin{lemma}\label{lm:Phi_sphere}We have\begin{align}&\max_{\lambda\in\sigma^{\Phi}(\hat{\mathscr G})}|\R\lambda|\leq\Vert\R\hat {\mathscr G}\Vert_{\Phi(O(\mathbf 0,R);\mathbb C^3)}\notag\\\leq{}&\min\left\{\frac{1}{2}+ \frac{3(kR)^{2} }{5\pi}\left(\frac{4\pi}{3}\right)^{2/3},2+\frac{4}{5} kR+\frac{1}{2\pi} \left(  \sqrt[3]{k R+\frac{1}{2}}+\frac{4}{9 \sqrt[3]{k R+\frac{1}{2}}} \right)\right\}.
\tag{\ref{eq:Re_gamma_norm_bd}$'$}\label{eq:Re_gamma_norm_bd'}
\end{align}\end{lemma}\begin{proof}To establish the bound in \eqref{eq:Re_gamma_norm_bd'}, we need to show that   $\Vert\hat{ \mathscr G}-\hat{ \gamma}\Vert_{\Phi(O(\mathbf0,R);\mathbb C^3)}=\frac12$, which sharpens $ \Vert\hat{ \mathscr G}-\hat{ \gamma}\Vert_{L^{2}(V;\mathbb C^3)}\leq1$. Here, we note that $\bm F_{\ell m}(\bm r)=\bm F_{\ell m}(|\bm r|,\theta,\phi)=\nabla[|\bm r|^{\ell}Y_{\ell m}(\theta,\phi)] $ satisfies the the following eigenequation
\begin{align}&(( \hat \gamma-\hat {\mathscr G})\bm F_{\ell m})(\bm r)=\nabla\varoiint_{\partial O(\mathbf 0,R)}\frac{\bm n'\cdot\bm F_{\ell m}(\bm r')}{4\pi|\bm r-\bm r'|}\D S'\notag\\={}&\nabla\sum_{\ell'=0}^\infty\sum_{m'=-\ell'}^{\ell'}\varoiint_{\partial O(\mathbf 0,R)}\frac{\bm n'\cdot\bm F_{\ell m}(\bm r')|\bm r|^{\ell}Y_{\ell' m'}^{*}(\theta',\phi')Y_{\ell' m'}^{\phantom{*}}(\theta,\phi)}{(2\ell+1)|\bm r'|^{\ell+1}}\D S'=\frac{\ell\bm F_{\ell m}(\bm r)}{2\ell+1},\end{align}and we consequently have the spectral decomposition:\begin{align}(\hat{\mathscr G}-\hat\gamma)\bm E=-\sum_{\ell=1}^\infty\sum_{m=-\ell}^\ell\frac{\ell\langle \bm e_{\ell m},\bm E\rangle_{O(\mathbf0,R)}\bm e_{\ell m}}{2\ell+1},\quad\forall \bm E\in\Phi(O(\mathbf0,R);\mathbb C^3), \label{eq:spec_decomp_G_gamma_sphere}\end{align}with the complete orthonormal basis set $\{\bm e_{\ell m}(\bm r)=(\ell R^{2\ell +1})^{-1/2}\nabla(|\bm r|^\ell Y_{\ell m}(\theta,\phi))|\ell\in\mathbb Z_{>0},m=\mathbb Z\cap[-\ell,\ell]\} $ for $ \Cl((\hat{\mathscr G}-\hat\gamma)\Phi(O(\mathbf0,R);\mathbb C^3))$,
so $ \Vert\hat{ \mathscr G}-\hat{ \gamma}\Vert_{\Phi(O(\mathbf0,R);\mathbb C^3)}=\sup_{\ell\in\mathbb Z_{\geq1}}\frac{\ell}{2\ell+1}=\frac12$.
 \end{proof}

A natural consequence of Theorem \ref{thm:L2_op_norm} and Lemma \ref{lm:Phi_sphere} is a set of inequalities concerning the explicitly defined Mie resonance modes.

\begin{corollary}[$L^2$-bounds for Mie resonances]\label{cor:L2_Mie}Set $x=kR$, and consider  the resonance modes $ \chi_{\ell,x}^{\mathrm {TE}}$ and $ \chi_{\ell,x}^{\mathrm {TM}}$ that are  the roots of $ \Delta_{\ell,x}^{\mathrm {TE}}(\chi)$ and $\Delta_{\ell,x}^{\mathrm {TM}}(\chi)$ [defined in \eqref{eq:TE_res}--\eqref{eq:TM_res}]. The expressions
\begin{small}\begin{align}
Q_R(L,x)\colonequals{}&\frac{\min\left\{\frac{1}{2}+\frac{3x^2 }{5\pi}\left(\frac{4\pi}{3}\right)^{2/3},2+\frac{4x}{5}+\frac{1}{2\pi} \left(  \sqrt[3]{x+\frac{1}{2}}+\frac{4}{9 \sqrt[3]{x+\frac{1}{2}}} \right)\right\}}{\max_{\ell\in\mathbb Z\cap[1,L]}\max\{|\R(1/\chi_{\ell,x}^{\mathrm {TE}})|,|\R(1/\chi_{\ell,x}^{\mathrm {TM}})|\}},\\Q_I(L,x)\colonequals{}&\frac{\min\left\{\frac{11x^{3}}{45}, \frac{9x}{16}\right\}}{\max_{\ell\in\mathbb Z\cap[1,L]}\max\{|\I(1/\chi_{\ell,x}^{\mathrm {TE}})|,|\I(1/\chi_{\ell,x}^{\mathrm {TM}})|\}}
\end{align}\end{small}are both greater than or equal to unity, for each $L\in\mathbb Z_{>0}$. \end{corollary}

 For every fixed value of $ x=kR$, by finding the roots of  \eqref{eq:TE_res}--\eqref{eq:TM_res} numerically, we can check the viability of the aforementioned inequalities against some specific examples (Fig.~\ref{fig:2-1}).

 \begin{figure}[h]
\begin{minipage}{.5\textwidth}\begin{tikzpicture}\pgfplotsset{xlabel style={yshift=.3cm}, ylabel style={yshift=-0.6cm},width=7.2cm,height=4cm,xmajorgrids,ymajorgrids,xminorgrids,yminorgrids, tick label style={font=\tiny},label style={font=\tiny}}
\begin{loglogaxis}[xticklabels={0.01,0.1,1,10},yticklabels={0.01,0.1,1,10},xlabel={$x=kR$},ylabel={$Q_R(\max\{5,2x\},x)-1$},ymin=0.05,ymax=10,xmin=0.045,xmax=55,enlargelimits=false,legend style={
legend columns=1,
cells={anchor=west},at={(.98,.95)},
font=\tiny,legend style={row sep=-2.5pt},
}]\addplot [only marks,
draw=red,mark=o,thin,mark size=1pt
] coordinates{(0.05, 0.102786)(0.0561009, 0.103508)(0.0629463, 0.104416)(0.0706269, 0.105559)(0.0792447, 0.106999)(0.088914, 0.108811)(0.0997631, 0.111093)(0.111936, 0.113967)(0.125594, 0.117584)(0.140919, 0.122139)(0.158114, 0.127874)(0.177407, 0.135096)(0.199054, 0.144191)(0.223342, 0.155645)(0.250594, 0.170072)(0.281171, 0.188246)(0.315479, 0.211144)(0.353973, 0.239998)(0.397164, 0.276368)(0.445625, 0.322227)(0.5, 0.380074)(0.561009, 0.453082)(0.629463, 0.545284)(0.706269, 0.661826)(0.792447, 0.809295)(0.88914, 0.99616)(0.997631, 1.23338)(1.11936, 1.53523)(1.25594, 1.92053)(1.40919, 2.41441)(1.58114, 3.05112)(1.77407, 3.87848)(1.99054, 4.01851)(2.23342, 4.40395)(2.50594, 4.68702)(2.81171, 4.30398)(3.15479, 3.7064)(3.53973, 3.38544)(3.97164, 3.37337)(4.45625, 3.60048)(5., 3.00504)(5.61009, 3.02008)(6.29463, 3.19824)(7.06269, 3.23345)(7.92447, 3.25727)(8.8914, 3.29554)(9.97631, 3.41031)(11.1936, 3.5865)(12.5594, 3.68241)(14.0919, 3.73746)(15.8114, 3.90919)(17.7407, 3.92657)(19.9054, 3.98055)(22.3342, 4.14202)(25.0594, 4.29047)(28.1171, 4.31291)(31.5479, 4.47009)(35.3973, 4.74055)(39.7164, 4.80159)(44.5625, 4.96701)(50., 4.98617)};

\end{loglogaxis}
\end{tikzpicture}
\vspace{-1em}
\begin{center}\begin{tiny}\quad\;\;\;(a)\end{tiny}\end{center}\end{minipage}\begin{minipage}{.5\textwidth}\begin{tikzpicture}\pgfplotsset{xlabel style={yshift=.3cm}, ylabel style={yshift=-.6cm},width=7.2cm,height=4cm,xmajorgrids,ymajorgrids,xminorgrids,yminorgrids, tick label style={font=\tiny},label style={font=\tiny}}
\begin{loglogaxis}[xticklabels={0.01,0.1,1,10},yticklabels={0.01,0.1,1,10},xlabel={$x=kR$},ylabel={$Q_I(\max\{5,2x\},x)-1$},ymin=0.05,ymax=10,xmin=0.045,xmax=55,enlargelimits=false, ,legend style={
legend columns=1,
cells={anchor=west},at={(.98,.95)},
font=\tiny,legend style={row sep=-2.5pt},
}]\addplot [only marks,
draw=red,mark=o,thin,mark size=1pt
] coordinates{(0.05, 0.100704)(0.0561009, 0.100762)(0.0629463, 0.100873)(0.0706269, 0.101103)(0.0792447, 0.101348)(0.088914, 0.101735)(0.0997631, 0.102182)(0.111936, 0.102766)(0.125594, 0.103479)(0.140919, 0.104381)(0.158114, 0.105538)(0.177407, 0.106975)(0.199054, 0.1088)(0.223342, 0.111107)(0.250594, 0.114028)(0.281171, 0.117732)(0.315479, 0.122446)(0.353973, 0.12846)(0.397164, 0.136169)(0.445625, 0.146107)(0.5, 0.159031)(0.561009, 0.176024)(0.629463, 0.198681)(0.706269, 0.229329)(0.792447, 0.271121)(0.88914, 0.32758)(0.997631, 0.40145)(1.11936, 0.494519)(1.25594, 0.610716)(1.40919, 0.760678)(1.58114, 0.808038)(1.77407, 0.64771)(1.99054, 0.55736)(2.23342, 0.541623)(2.50594, 0.616437)(2.81171, 0.365116)(3.15479, 0.189791)(3.53973, 0.137111)(3.97164, 0.222359)(4.45625, 0.28642)(5., 0.27128)(5.61009, 0.388828)(6.29463, 0.380975)(7.06269, 0.401012)(7.92447, 0.486076)(8.8914, 0.561548)(9.97631, 0.532139)(11.1936, 0.615008)(12.5594, 0.666956)(14.0919, 0.68911)(15.8114, 0.755336)(17.7407, 0.772902)(19.9054, 0.850158)(22.3342, 0.87504)(25.0594, 0.945279)(28.1171, 0.982523)(31.5479, 1.02846)(35.3973, 1.06148)(39.7164, 1.1162)(44.5625, 1.16399)(50., 1.23414)};

\end{loglogaxis}
\end{tikzpicture}
\vspace{-1em}
\begin{center}\begin{tiny}\qquad\;\;(b)\end{tiny}\end{center}\end{minipage}
\caption{Numerical validation of  Corollary \ref{cor:L2_Mie} for \emph{(a)}~the real part and \emph{(b)}~the imaginary part of  eigenvalues in Mie scattering. For numerical root searching, we run Newton's method with initial guesses $\chi\in\{e^{im\pi/8}|m\in\mathbb Z\cap[1,7]\}$.  \label{fig:2-1}}
\end{figure}

Next, we will check the bound estimate in Theorem  \ref{thm:HS_norm}   against a spherical scatterer of radius $R=x/k$, where $ \Vert\hat{\mathscr G }(\hat I+2\hat{\mathscr G})^2\Vert_2$ dominates \begin{align}
\mu_{\mathrm{HS}}(x)\colonequals\sqrt{\sum_{{\ell\in\mathbb Z_{>0},\Delta_{\ell,x}^{\mathrm{TE}}(1/\lambda)\Delta_{\ell,x}^{\mathrm{TM}}(1/\lambda)=0}}
(2\ell+1)|\lambda|^2|1+2\lambda|^4},\label{eq:muHS}\end{align}via a special case
of the first
inequality in \eqref{eq:conv_spectral_series}. Here, the factor $2\ell+1$ represents the degree of degeneracy for each eigenvalue associated with a $ Y_{\ell m}$ mode.
\begin{corollary}[Hilbert--Schmidt bounds for Mie resonances]\label{cor:HS_Mie}Set $ \zeta(s)\colonequals\sum_{n=1}^\infty n^{-s}$ for $s>1$. Define \begin{align}
M_{\mathrm{HS}}(x)\colonequals\sqrt{\frac{7 \zeta (3)}{32}+\frac{31 \zeta (5)}{128}-\frac{\pi ^4}{192}}+\frac{3x^{2}}{2}\{3+4B_{\gamma}(x)+2[B_{\gamma}(x)]^{2}\}.
\end{align}for \begin{small}\begin{align}
B_{\gamma}(x)\colonequals{}&\min \left\{\frac{3x^{2}}{5 \pi } \left(\frac{4 \pi }{3}\right)^{2/3} ,2+\frac{4x}{5}+\frac{1}{2\pi} \left(  \sqrt[3]{x+\frac{1}{2}}+\frac{4}{9 \sqrt[3]{x+\frac{1}{2}}} \right)\right\}+\min \left\{\frac{11x^{3}}{45}, \frac{9x}{16}\right\}.
\end{align}\end{small}We have $ \mu_{\mathrm{HS}}(x)\leq M_{\mathrm{HS}}(x)$ for all $x>0$.
\end{corollary}
\begin{proof}
According to \eqref{eq:spec_decomp_G_gamma_sphere}, we have\begin{align}\Vert(\hat I+2\hat{\mathscr G}-2\hat\gamma)^2(\hat{\mathscr G}-\hat\gamma)\Vert_2=\sqrt{\sum_{\ell=1}^\infty\frac{\ell^2}{(2\ell+1)^{5}}}=\sqrt{\frac{7 \zeta (3)}{32}+\frac{31 \zeta (5)}{128}-\frac{\pi ^4}{192}}.\end{align} Therefore, we obtain \begin{align}\Vert(\hat I+2\hat{\mathscr G})^2\hat{\mathscr G }\Vert_2\leq{}&\Vert(\hat I+2\hat{\mathscr G}-2\hat\gamma)^2(\hat{\mathscr G}-\hat\gamma)\Vert_2+\Vert\hat \gamma+4[(\hat{\mathscr G}-\hat\gamma)\hat \gamma+\hat\gamma(\hat{\mathscr G}-\hat\gamma)](\hat I+\hat{\mathscr G}-\hat\gamma)+\notag\\&+4(\hat{\mathscr G}-\hat\gamma)^{2}\hat \gamma+4(\hat I+\hat{\mathscr G}-\hat\gamma)\hat \gamma^2+4\hat\gamma(\hat{\mathscr G}-\hat\gamma)\hat\gamma+4\hat \gamma^2(\hat{\mathscr G}-\hat\gamma)+4\hat\gamma^3\Vert_2\notag\\\leq{}&\sqrt{\frac{7 \zeta (3)}{32}+\frac{31 \zeta (5)}{128}-\frac{\pi ^4}{192}}+\Vert\hat\gamma\Vert_2(6+8\Vert\hat\gamma\Vert+4\Vert\hat\gamma\Vert^2),\end{align}via $ \Vert\hat{\mathscr G}-\hat\gamma\Vert_{\Phi(O(\mathbf0,R);\mathbb C^3)}=\frac12$ and $ \Vert \hat I+\hat{\mathscr G}-\hat\gamma\Vert_{\Phi(O(\mathbf0,R);\mathbb C^3)}\leq1$.

Noting that Theorems  \ref{thm:L2_op_norm} and \ref{thm:HS_norm} bring us $\Vert\hat\gamma\Vert\leq B_\gamma(x)$
and
$\Vert\hat\gamma\Vert_2\leq\frac{3x^{2}}{4 }  $ respectively, we arrive at the claimed inequality through $ \mu_{\mathrm{HS}}(x)\leq \Vert(\hat I+2\hat{\mathscr G})^2\hat{\mathscr G }\Vert_2\leq M_{\mathrm{HS}}(x)$.
 \end{proof}

In Fig.~\ref{fig:2-2}, we numerically verify Corollary \ref{cor:HS_Mie} for some special values of $x=kR$.

\begin{figure}[h]
\begin{minipage}{0.6\textwidth}
\begin{tikzpicture}\pgfplotsset{xlabel style={yshift=.3cm}, ylabel style={yshift=-0.4cm},width=8cm,height=3.5cm,xmajorgrids,ymajorgrids,xminorgrids,yminorgrids, tick label style={font=\tiny},label style={font=\tiny}}
\begin{loglogaxis}[xticklabels={0.01,0.1,1,10},yticklabels={0.01,0.1,1,10},xlabel={$x=kR$},ylabel={$\sqrt{ \frac{M_{\mathrm{HS}}(x)}{\mu_{\mathrm{HS}}(x)}}-1$},ymin=0.05,ymax=12,xmin=0.045,xmax=55,enlargelimits=false, minor x tick num =4,xtick={0.01,0.1,1,10,100},ytick={.01,.1,1,10,100},legend style={
legend columns=1,
cells={anchor=west},at={(.98,.95)},
font=\tiny,legend style={row sep=-2.5pt},
}]\addplot [only marks,
draw=red,mark=o,thin,mark size=1pt
] coordinates{(0.05, 0.0797108)(0.0561009, 0.0958267)(0.0629463, 0.115748)(0.0706269, 0.140277)(0.0792447, 0.170338)(0.088914, 0.206984)(0.0997631, 0.251386)(0.111936, 0.304821)(0.125594, 0.368649)(0.140919, 0.444271)(0.158114, 0.533086)(0.177407, 0.636428)(0.199054, 0.755479)(0.223342, 0.891168)(0.250594, 1.04403)(0.281171, 1.21402)(0.315479, 1.40034)(0.353973, 1.60119)(0.397164, 1.81363)(0.445625, 2.03349)(0.5, 2.25544)(0.561009, 2.47309)(0.629463, 2.679)(0.706269, 2.86459)(0.792447, 3.02041)(0.88914, 3.13764)(0.997631, 3.21099)(1.11936, 3.24144)(1.25594, 3.23841)(1.40919, 3.22044)(1.58114, 3.19527)(1.77407, 3.22061)(1.99054, 3.29924)(2.23342, 3.42076)(2.50594, 3.64846)(2.81171, 3.9788)(3.15479, 4.33069)(3.53973, 4.16406)(3.97164, 4.05517)(4.45625, 3.96407)(5., 3.89814)(5.61009, 3.87115)(6.29463, 3.85486)(7.06269, 3.85937)(7.92447, 3.87813)(8.8914, 3.91256)(9.97631, 3.95479)(11.1936, 4.01056)(12.5594, 4.07661)(14.0919, 4.15118)(15.8114, 4.23448)(17.7407, 4.32593)(19.9054, 4.43191)(22.3342, 4.54811)(25.0594, 4.79077)(28.1171, 4.96708)(31.5479, 5.2972)(35.3973, 5.73528)(39.7164, 6.34739)(44.5625, 6.69998)(50., 6.99368)};
\end{loglogaxis}
\end{tikzpicture}
\end{minipage}\begin{minipage}{0.4\textwidth}
\caption{Numerical validation of $\mu_{\mathrm{HS}}(x)\leq M_{\mathrm{HS}}(x)$. For truncated approximations to the infinite sum $\mu_{\mathrm{HS}}(x)$, we perform root searching as in  Fig.~\ref{fig:2-1} for each  $ \ell\in\mathbb Z\cap[1,\max\{5,2x\}]$ and  merge numerical eigenvalues $ \lambda=1/\chi$ that agree up to $8$ decimal places. \label{fig:2-2}}
\end{minipage}
\end{figure}

\section{Discussions}

In this work, we have carried out error analysis for perturbative solutions to the Born equation, with norm bound estimates based on geometric shapes of dielectric scatterers.
By  numerical search for the eigenvalues in Mie scattering for spheres of various radii, we have found that the norm bounds (for spherical scatterers) computed in \S\ref{sec:norm_bds}  overestimate the actual value by a moderate factor, but still within the same order of magnitude (see Figs.\ \ref{fig:2-1} and \ref{fig:2-2}). Perhaps one can obtain a tighter bound on the operator norms by refining the analysis in \S\ref{sec:norm_bds}. It is also possible to improve our estimates for non-spherical scatterers by taking into account more geometric information, other  than the radius of the circumsphere.
\subsection*{Acknowledgments}This research was supported in part  by the Applied Mathematics Program within the Department of Energy
(DOE) Office of Advanced Scientific Computing Research (ASCR) as part of the Collaboratory on
Mathematics for Mesoscopic Modeling of Materials (CM4).

The author thanks Prof.~Xiaowei Zhuang (Harvard University) for her  questions on nanophotonics   in 2006, which inspired the current work. Part of this research formed   Chapters 4, 5 and Appendices C, D in the  author's  PhD thesis \cite{ZhouThesis} completed in January 2010 under her supervision. The author dedicates this paper to her 50th birthday.


\end{document}